\documentclass[a4paper,UKenglish]{lipics}

\usepackage{verbatim}

\usepackage{tikz}

\usepackage[utf8]{inputenc}
\usepackage[T1]{fontenc}

\usepackage{wrapfig} 
\usepackage{xspace}
\newcommand{\para}[1]{\smallskip\noindent\emph{#1}.\xspace}

\usepackage{amsthm,amssymb,amsmath}

\theoremstyle{definition}
\renewcommand{\proofname}{Proof}
\theoremstyle{definition}

\newtheorem{prop}{Proposition}
\newtheorem{defi}{Definition}
\newtheorem{lem}{Lemma}
\newtheorem{corol}{Corollary}
\newtheorem*{rmk}{Remark}

\newtheorem*{addc}{Additional Constraints}

\usepackage{stmaryrd}
\usepackage{ebproofs}

\newcommand{\fnsz}{\footnotesize}

\newcommand{\ems}{[\;]}
\newcommand{\eft}{(\;)}
\newcommand{\rew}{\rightarrow}
\newcommand{\al}{\alpha}
\newcommand{\gam}{\gamma}
\newcommand{\Gam}{\Gamma}
\newcommand{\sig}{\sigma}
\newcommand{\epsi}{\varepsilon}
\newcommand{\fom}{f^{\omega}}
\newcommand{\arewa}{[\alpha]\rew \alpha}
\newcommand{\erewa}{[\;]\rew \alpha}
\newcommand{\ovl}[1]{\overline{#1}}

\newcommand{\rstr}[1]{|_{#1}}
\newcommand{\nr}{\tt nr}
\newcommand{\arob}{\symbol{64}}

\newcommand{\TermV}{\mathscr{V}}

\newcommand{\cu}{ {\tt Y}}
\newcommand{\ttS}{ \tt S }
\newcommand{\bbN}{\mathbb{N}}
\newcommand{\scrM}{\mathscr{M}}
\newcommand{\scrMo}{\mathscr{M}_0}
\newcommand{\scrV}{\mathscr{V}}
\DeclareMathOperator{\supp}{\tt supp}
\DeclareMathOperator{\Types}{\tt Types}
\DeclareMathOperator{\FTypes}{\tt FTypes}
\DeclareMathOperator{\Deriv}{\tt Deriv}

\newcommand{\EFO}{\tt EFO }

\DeclareMathOperator{\Rep}{\tt Rep}

\DeclareMathOperator{\leqfty}{\leqslant_\infty}
\DeclareMathOperator{\supf}{\,\!^f\!}
\DeclareMathOperator{\supo}{\,\!^0\!}

\DeclareMathOperator{\init}{\tt init}
\DeclareMathOperator{\dom}{\tt dom}

\DeclareMathOperator{\Ax}{\tt Ax}
\DeclareMathOperator{\Hd}{\tt Hd}

\DeclareMathOperator{\Tl}{\tt Tl}

\DeclareMathOperator{\Rt}{\tt Rt}

\DeclareMathOperator{\Res}{\tt Res}

\DeclareMathOperator{\codom}{\tt codom}
\DeclareMathOperator{\depth}{\tt d}

\DeclareMathOperator{\cd}{\tt cd}
\DeclareMathOperator{\spo}{\tt sp}
\DeclareMathOperator{\sip}{\tt sip}
\DeclareMathOperator{\lcp}{\tt lcp}
\DeclareMathOperator{\rdeg}{\tt rdeg}

\DeclareMathOperator{\bip}{{\tt b}}
\DeclareMathOperator{\bisupp}{\tt bisupp}

\DeclareMathOperator{\lra}{\leftrightarrow}

\DeclareMathOperator{\Approx}{\tt Approx}
\DeclareMathOperator{\Reach}{\tt Reach}

\DeclareMathOperator{\tr}{\tt tr}
\DeclareMathOperator{\AxTr}{\tt AxTr}
\DeclareMathOperator{\OutTr}{\tt OutTr}
\DeclareMathOperator{\RedTr}{\tt RedTr}

\DeclareMathOperator{\Arg}{\tt Arg}
\DeclareMathOperator{\pos}{\tt pos}

\newcommand{\ax}{\tt ax}
\newcommand{\abs}{\tt abs}
\newcommand{\app}{\tt app}

\DeclareMathOperator{\Rft}{R}

\DeclareMathOperator{\eqm}{\stackrel{\mathscr{M}}{\equiv}}

\newcommand{\rewb}[1]{\stackrel{ #1}{\rightarrow}}
\newcommand{\code}[1]{\lfloor #1 \rfloor}
\newcommand{\red}[1]{\textcolor{red}{#1}}
\newcommand{\blue}[1]{\textcolor{blue}{#1}}
\newcommand{\ad}[1]{\text{ad}(#1)}

\newcommand{\up}{\tt up}
\newcommand{\topb}{\tt top}

\newcommand{\trck}[1]{~ \red{\text{(tr. }#1\text{)}}}
\newcommand{\posPr}[1]{~ \red{\text{(at\,}#1\text{)}}}

\newcommand{\fP}{\supf P}

\title{Infinitary Intersection Types as Sequences: a New Answer to Klop's Question}

\author{Pierre Vial}

\affil{IRIF, Université Paris-Diderot 
\texttt{pvial@pps.univ-paris-diderot.fr}}

\authorrunning{Pierre~Vial}
\Copyright{Pierre Vial}

\begin{document}

\maketitle


\begin{abstract}
 We provide a type-theoretical characterization of weakly-normalizing 
terms in an infinitary lambda-calculus. We adapt for this purpose the 
standard  quantitative (with non-idempotent intersections) type assignment 
system of the lambda-calculus to our infinite calculus.

Our work provides a new answer to Klop's HHN-problem, namely, finding 
out if there is a type system characterizing the hereditary 
head-normalizing (HHN) lambda-terms. Tatsuta showed that HHN could not be 
characterized by a finite type system. We prove that an infinitary type 
system endowed with a validity condition called approximability can 
achieve it. 


\end{abstract}


\section{Introduction}

The \textbf{head-normalizing (HN) terms} can be characterized by
various \textit{intersection} type systems. Recall  that  a term is HN
if it can be reduced to a \textbf{head-normal form (HNF)},
\textit{i.e.} a term $t$ of the form $\lambda x_1\ldots
x_p.(x\,t_1)\ldots t_q\ (p \geq0, q \geq 0)$, where $x$ is referred as the
\textbf{head-variable} of $t$ and the terms $t_1,\ldots,\,t_q$ as the
\textbf{arguments} of the head-variable $x$.

In general, intersection type frameworks, introduced by Coppo and
Dezani \cite{DBLP:journals/ndjfl/CoppoD80}, allow to characterize many
classes of normalizing terms, such as the \textbf{weakly
normalizing (WN) terms} (see \cite{DBLP:journals/tcs/Bakel95} for an
extensive survey).  A term is WN if it can be reduced to a
\textbf{normal form (NF)}, \textit{i.e.} a term without
\textit{redexes}. \textit{Inductively}, a term is WN if it is HN and all the
arguments of its head-variable are WN (it is meant that the base 
cases of this induction are the terms whose HNF is $\lambda
x_1\ldots x_p.x$).

According to Tatsuta \cite{DBLP:conf/flops/Tatsuta08}, the question of
finding out a type system characterizing \textbf{hereditary
 head-normalizing (HHN) terms} was raised by Klop in a private
exchange with Dezani in the late 90s. The definition of  HHN term
is given by the \textit{coinductive} version of the above inductive
definition: \textit{coinductively}, a term is HHN if it is HN and all
the arguments of its head variable are themselves HHN. It is
equivalent to say that the B\"ohm tree of the term does not hold any
occurrence of $\bot$. Tatsuta focused his study on \textit{finitary}
type systems and showed Klop's problem's answer was negative for them,
by noticing that the set of HHN terms was not recursively enumerable.

Parallelly, the B\"ohm trees without $\bot$ can be seen as the set of
normal forms of an infinitary calculus, referred as $\Lambda^{001}$ in
\cite{DBLP:journals/tcs/KennawayKSV97}, which has been reformulated
very elegantly in  coinductive frameworks~\cite{DBLP:conf/rta/EndrullisHHP015,DBLP:conf/rta/Czajka14}. In this calculus, the HHN terms
correspond to the infinitary variant of the WN terms. An infinite term
is WN if it can be reduced to a NF by at least one \textbf{strongly
  converging reduction sequence (s.c.r.s.)}, which constitute a
special kind of reduction sequence of (possibly) infinite length,
regarded as \textit{sound}. This motivates to check whether an
\textit{infinitary} type system is able to characterize HHN terms in
the infinite calculus $\Lambda^{001}$.

We use a quantitative, resource-aware type system to help us achieve 
this goal. In those type systems, typability is known to imply 
normalizability by a very simple (variant of the same) argument. Namely, 
reducing a typed redex inside a derivation decrease some 
non-negative integer measure, which entails that the reduction must stop at 
some point (see for instance \cite{DBLP:conf/ifipTCS/BucciarelliKR14} or 
Lemma \ref{lemTypHN}).  This is unlikely to be adapted in an infinitary 
framework. However, quantitative type derivations do have very simple 
and readable combinatorial features that will turn out to be useful to
build an infinitary type system. In particular, reduction inside a 
derivation almost comes down to \textit{moving} parts of the original 
derivation, without adding new rules (a figure is given in \S~
\ref{subsecResBip}).

\subsubsection*{Contributions}

We define an infinitary quantitative type system, inspired by
the finitary de Carvalho's system $\mathscr{M}_0$ \cite{decarvalho07phd}. 
However, we show that a direct coinductive adaptation of system $\mathscr{M}_0$
cannot work for two reasons (\textit{c.f.} Section \ref{secExamples}):
\begin{itemize}
\item It would lead to the possibility of typing some non-HN terms, like
$\Delta\Delta$. That is why a validity criterion is needed to
discard irrelevant derivations, as in other infinitary frameworks
\cite{santocanale01brics}.
\item This validity criterion relies on the idea of \textbf{approximability}.
It can be seen that multisets are not fit to formally express such a notion, which motivates the need for rigid constructions: multisets of types
are replaced (coinductively) by families of types indexed by integers
called \textbf{tracks}.
\end{itemize}
Tracks constitute the main feature of the type system presented here. 
They act as \textit{identifiers} and allow us to bring out a
combinatorics that existed implicitly -- but could not be formulated 
-- in regular quantitative type systems, where multiset constructions 
made it impossible to distinguish two copies of the same type. For instance,
we will be able to trace any type through the rules of a whole typing
derivation. Our framework is deterministic, \textit{e.g.} there is a \textit{unique} way to produce a derivation from another one while reducing a redex.

\subsubsection*{Outline}

We informally discuss the necessity of the notion of approximability and rigid
constructions in \S~\ref{secExamples}. In \S~\ref{secRigid},
we formally define our terms, type system and tracks. In
\S~\ref{secDynamics}, we define reduction and expansion of a
typing derivation, as well as \textit{residuals}. In
\S~\ref{secApproxUF}, we formulate the approximability condition
and the WN characterization criterion (called \textit{unforgetfulness} here), and next, we prove an infinitary subject reduction property. In
\S~\ref{secNF}, we describe all the sound derivations typing a
normal form and  prove an infinitary subject expansion
property. It concludes the proof of our type-theoretic characterization of WN.


\section{Informal Discussion}

\label{secExamples}

In this section, we \textbf{informally} introduce, through a few
examples, the key concepts of our work, namely \textit{rigidity} and
\textit{approximability}.

\subsection{The Finitary Type System $\scrMo$ and Unforgetfulness}
\label{subsecMo}

Let us first recall the typing system $\scrMo$ with
\textit{non-idempotent} intersection types ~\cite{decarvalho07phd},
given by the following \textit{inductive} grammar $\sigma,\,\tau :: = \alpha~
|~ [\sigma_i]_{i\in I} \rightarrow \tau $, where the constructor
$\ems$ is used for finite multisets, and the type variable $\alpha$
ranges over a countable set $\mathscr{X}$ of type variables.  We write
$[\sigma]_n$ to denote the multiset containing $\sigma$ with
multiplicity $n$. The multiset $[\sigma_i]_{i\in I}$ is meant to be the
intersection of the types $\sigma_i$, taking into account their
\textit{multiplicity}. In idempotent intersection type systems, the
type intersections $A\wedge B\wedge A$ and $A\wedge B$ are \textit{de facto} equal, whereas in $\scrMo$, the multiset types $[\sigma,\,\tau,\,\sigma]$ and $[\sigma,\,\tau]$ are not. No weakening is allowed either, \textit{e.g.} $\lambda x.x$ can be typed with $[\tau]\rew \tau$, but \textit{not} with $[\tau,\sigma]\rew \tau$.

In system $\scrMo$, a \textit{judgment} is a triple $\Gam \vdash t:\,\sigma$, where $\Gam$ is a context, \textit{i.e.} a function from the set $\scrV$ of term variables to the sets of multiset types $[\sigma_i]_{i\in I}$, $t$ is a term and $\sigma$ is a type. The multiset union + is extended point-wise on contexts. Let us consider the rules below:

\begin{center}
\begin{prooftree}
\Infer{0}[ax]{x:\, [\tau] \vdash x:\,\tau }
\end{prooftree}
 \hspace{3cm}
\begin{prooftree}
\Hypo{\Gamma,\,x:\,[\sigma_i]_{i\in I} \vdash t:\, \tau }
\Infer{1}[abs]{\Gamma \vdash \lambda x.t:~
[\sigma_i]_{i\in I} \rightarrow \tau}
\end{prooftree}\\[0.3cm]
\begin{prooftree}
\Hypo{\Gamma \vdash t:\, [\sigma_i]_{i\in I}\rightarrow \tau
 }
\Hypo{\Delta_i \vdash u:\, \sigma_i}
\Delims{ \left( }{ \right )^{i\in I} }
\Infer{2}[app]{\Gamma + \sum\limits_{i \in I} \Delta_i \vdash t(u):\, \tau
} 
\end{prooftree}
\begin{prooftree}
\Hypo{\Pi_k }
\Hypo{(\Pi_k)_{i \in I} }
\Infer{2}[app]
      {\Delta_i \vdash t(u)}
\end{prooftree}

\end{center}

The set of derivations is defined \textit{inductively} by the above rules.  We write $\Pi \rhd \Gamma \vdash t:\, \tau$ to mean that the (finite) derivation $\Pi$ concludes with the judgment $\Gamma \vdash t:\,\tau$.
A term is HN iff it is typable in system $\scrMo$.

\begin{rmk} 
  The rule $\app$ relies on the equality between two multisets: the multisets of the types typing $u$ and the negative part of the arrow type typing $t$ must be equal to grant that $tu$ is typable. In constrast to equality between two sequences, the multiset equality $[\sigma_i]_{i\in I}=[\sig'_i]_{i\in I'}$ can be seen as \textit{commutative} since the order we use to list the elements of the involved m.s. is of no matter (it is intuitively \textit{collapsed} for m.s.).
\end{rmk}

Notice that if $x$ is assigned $\ems\rew \tau$, then $x\,t$ is typable with type $\tau$ for any term $t$ -- which is left untyped -- even if $t$ is not HN. In order to characterize WN, we must grant somehow that every subterm (that cannot be erased during a reduction sequence) is typed : $\ems$ should not occur at bad positions in a derivation $\Pi$. Actually, it is enough to only look at the judgment concluding $\Pi$ : a term $t$ is WN iff it is typable in $\scrMo$ inside an
\textbf{unforgetful} judgment. We say here that judgment $\Gamma \vdash t:\,\tau$ is unforgetful when $\Gamma$ (resp. $\tau$) does not
hold negative (resp. positive) occurrences of $\ems$. The proper
definitions are to be found in \S~\ref{subsecUF}, but, for now, it
is enough for now to notice that a sufficient condition of unforgetulness
is to be \textbf{$\ems$-free}: $t$ is WN as soon as $\Gamma$ and $\tau$ do not
hold $\ems$.

\subsection{Infinitary Subject Expansion by Means of Truncation}

\label{subsecExamples}

Let us just admit that there is an infinitary version of $\scrMo$, that we call $\scrM$. System $\scrM$ allows infinite multiset (\textit{e.g.} $[\alpha]_{\omega}$ is the multiset in which $\alpha$ occurs with an infinite multiplicty, s.t. $[\alpha]_{\omega}=[\alpha]+[\alpha]_{\omega}$) and proofs of infinite depth.

Let $\Delta_f=\lambda x.f(xx)$ and $\cu=\Delta_f\Delta_f$. Notice $\cu\rew f(\cu)$, so $\cu \rew^{k} f^k(\cu)$. Intuitively, if $k\rew \infty$, the redex disapper and  we get $\cu \rew^{\infty} \fom$ where $\fom$ is the (infinite) term $f(f(f(\ldots)))$, satisfying $\fom=f(\fom)$ and containing a rightward infinite branch. Since $\fom$ does not hold any redex, $\fom$ can be seen as the NF of $\cu$.

In order to adapt the previous criterion, the idea is to type NF (here, $\fom$) in unforgetful judgment, and then proceed by (possibly infinite) expansion to get a typing derivation of the expanded term (here, $\cu$). Let us consider the following $\scrM$-derivation $\Pi'$ (presented as fixpoint):
\begin{center}
$\Pi'=$\begin{prooftree}
\Infer{0}[\text{ax}]{f:\,[[\alpha]\rightarrow \alpha] \vdash
f:\,[\alpha]\rightarrow \alpha}
\Hypo{\Pi'}
\Infer{1}{f:[[\alpha]\rightarrow \alpha]_{\omega} \vdash
f^{\omega}:\,\alpha}
\Infer{2}[app]{f:[[\alpha]\rightarrow \alpha]_{\omega} \vdash f^{\omega}:\,\alpha}
\end{prooftree}
\end{center}

Now, $\Pi'$ yields an unforgetful typing of $\fom$ (no occurrence of $\ems$). This the kind of derivation we want to expand in order to get a derivation $\Pi$ typing $\cu$. Since $\cu \rew^{\infty} \fom$ (infinite number of reduction steps), we are stuck. But notice that $\Pi'$ can be truncated into the derivation $\Pi'_n$ below for any $n\geqslant 1$ (we write  $\Gamma_n$ for $f:[\arewa]_{n-1}+[\erewa]$):
\begin{center}

  \begin{prooftree}

    \Infer{0}[\ax]{f:[\arewa]\vdash f:\arewa }

    \Infer{0}[\ax]{f:[\arewa]\vdash f:\arewa }

    \Infer{0}[\ax]{\Gamma_1\vdash f:\erewa }
    \Infer{1}[\app]{\Gamma_1 \vdash \fom:\alpha}     
    \Infer{2}[\app]{\Gamma_2 \vdash \fom :\alpha}
    \Ellipsis{}{\Gamma_{n-1}\vdash \fom :\alpha}
    \Infer{2}[\app]{\Gamma_n \vdash \fom : \alpha}

  \end{prooftree}
\end{center}

By \textbf{truncation}, we mean that the finite derivation $\Pi_n '$ can be (informally) obtained from the infinite one $\Pi'$ by erasing some elements from the infinite multisets appearing in the derivation. Conversely, we see that $\Pi'$ is the graphical \textbf{join} of the $\Pi'_n$: $\Pi'$ is obtained by superposing all the derivations $\Pi'_n$ on the same (infinite) sheet of paper.


However, we are still stuck: we do not know how to expand $\Pi'_n$, because although finite, it still types the $\infty$-reduced term $\fom$. But notice that we can replace $\fom$ by $f^k(\cu)$ inside $\Pi'_k$ whenever $k\geqslant n$, because those two terms do not differ in the typed parts of $\Pi'$ (\textbf{subject subsitution}). It yields a derivation $\Pi_n^k\rhd \Gam_n\vdash f^k(\cu):\alpha$. This time, $\Pi_n^k$ is a derivation typing the $k$-th reduced of $\cu$, so we can expand it $k$ times, yielding a derivation $\Pi_n$ ($\Pi_n$ does not depend on $k$). Then, we can rebuild a derivation $\Pi$ such that each $\Pi_n$ is a truncation of $\Pi$ the same way $\Pi'_n$ is of $\Pi'$ ($\Pi$ can be seen as the ``graphical'' join of the $\Pi_n$).

Thus, the ideas of truncation, subject subsitution and join indicate us how to perform $\infty$-subject expansion (\textit{cf.} \ref{}).
The particular form of $\Pi_n$ and $\Pi$ does not matter. Let us just say here that the $\Pi_n$ involve a family of types $(\gam)_{n\geqslant 1}$ inductively defined by $\gam_1=\erewa$ and $\gam_{n+1}=[\gam_i]_{1\leqslant i \leqslant n}\rew \alpha$ and $\Pi$ involves an infinite type $\gam$ satisfying $\gam=[\gam]_{\omega} \rew \alpha$.

Unfortunately, it is not difficult to see that the type $\gamma$ also
allows to type the non-HN term $\Delta \Delta$. Indeed,
$x:\,[\gamma]_{n\in \omega}\vdash xx:\,\alpha$ is derivable, so
$\vdash \Delta:\,\gamma$ and $\vdash \Delta\Delta:\, \alpha$ also are. 

This last observation shows that the naive extension of the standard non-idempotent type system to infinite terms is unsound as non-HN terms can be typed. Therefore, we need to discriminate between sound derivations (like $\Pi$ typing $\cu$) and unsound ones. For that, we define an infinitary derivation $\Pi$ to be \textbf{valid} or \textbf{approximable} when $\Pi$  admits finite truncations, generally denoted by $\supf \Pi$ -- that are  finite derivations of $\scrMo$ --, so that any fixed finite part of $\Pi$ is contained in some truncation $\supf \Pi$ (for now, a finite part of $\Pi$ informally denotes a finite selection of graphical symbols of $\Pi$, a formal definition is given in Sec.~\ref{subsecBisupp}).

\subsection{Safe Truncations of Typing Derivations}

\label{subsecNonD}

Truncating derivations  can obliterate 
different possible \textit{reduction choices} in system
$\scrM$. This problem suggests the need for rigid constructions. 

Let us consider a redex $t=(\lambda x.r)s$ and $t'=r[s/x]$. If a
derivation $\Pi$ types $t$ then $r$ has been given some type $\tau$ in
some context $\Gamma,\, x:\,[\sigma_i]_{i\in I}$ through a
subderivation $\Pi_0$ (see Appendix \ref{appSR} for a figure). Also, for each $i\in I$, $s$ has been given the
type $\sigma_i$ through some subderivation $\Pi_i$.  We can obtain a
derivation $\Pi'$ typing the term $t'$ by replacing the axiom rule
yielding $x:\,[\sigma_i]\vdash x:\sigma_i$ by the derivation
$\Pi_i$. The construction of such a $\Pi'$ from $\Pi$ relies generally
on a result referred as the "substitution lemma".

If a type $\sigma$ occurs several times in $[\sigma_i]_{i\in I}$ --
say $n$ times --, there must be $n$ axiom leaves in $\Pi$ typing $x$
with type $\sigma$, but also $n$ argument derivations $\Pi_i$ proving
$s:\,\sigma$.  When an axiom rule typing $x$ and an argument
derivation $\Pi_i$ are concluded with the same type $\sigma$, we shall
informally say that we can \textbf{associate} them. It means that this
axiom rule can be substituted by that argument derivation $\Pi_i$ when
we reduce $t$ to produce a derivation $\Pi'$ typing $t'$. There is not
only one way to associate the $\Pi_i$ to the axiom leaves typing $x$
(there can be as many as $n!$). Observe the following independent situations:
\begin{itemize}
\item Assume $\Pi_1$ and $\Pi_2$ (typing $s$), both concluded with the
  same type $\sigma= \sigma_1=\sigma_2$.  Thus, we also have two axiom leaves \#1 and \#2
  concluded by $x:\,[\sigma]\vdash x:\, \sigma$, where \#1 can be
  associated with $\Pi_1$ or $\Pi_2$.  When we truncate $\Pi$ into a
  finite $\supf \Pi$, the subderivation $\Pi_1$ and $\Pi_2$ are also
  cut into two derivations $\supf \Pi_1$ and $\supf \Pi_2$. In each
  $\supf \Pi_i$, $\sigma$ can be cut into a type $\supf \sigma_i$.  When 
  $\Pi_1$ and $\Pi_2$ are different, it is possible that $\supf
  \sigma_2 \neq \supf \sigma_1$ for every finite truncation 
   of $\Pi$. Thus, it is possible that, for every truncation
  $\supf \Pi$, the axiom leaf \#1 \textit{cannot}  be associated to $\supf
  \Pi_2$: 
indeed,  an association that is possible in $\Pi$ could be impossible
  for any of its truncations.
\item Assume this time $\sigma_1\neq \sigma_2$. When we truncate
$\Pi$ into a finite $\supf \Pi$, both $\sigma_1$ and $\sigma_2$ can 
be truncated into the same finite type $\supf \sigma$. We can then
associate $\supf \Pi_1$ with axiom \#2 and $\supf \Pi_2$ with axiom \#1
inside $\supf \Pi$, thus producing a derivation $\supf \Pi'$ typing 
$t'$ that has no meaning w.r.t. the possible associations 
in the original derivation $\Pi$.
\end{itemize}
That is why we will need a \textit{deterministic} association between the 
argument derivations and the axiom leaves typing $x$, so that the associations between them  are 
preserved even when we truncate derivations. System 
$\mathscr{M}$ does not allow to formulate a well-fit notion of 
approximability for derivations that would be stable under 
(anti)reduction and hereditary for subterms. This leads us to formulate
a rigid typing system in next section.

\section{A Rigid Type System}
\label{secRigid}



A non-negative integer is called here a \textbf{track}, an \textbf{argument track} is a integer $\geqslant 2$.
Let $\bbN^*$ the set of finite sequences of non-negative integers. If $a,\, b\in \bbN^*$, $a\cdot b$ is the concatenation of $a$ and $b$, $\epsi$ is the empty-sequence and $a\leqslant b$ if there is $c\in \bbN^*$ s.t. $b=a\cdot c$.
The length of $a$ is written $|a|$. The \textbf{applicative depth} $\ad{a}$ of $a\in \bbN^*$ is the number of argument tracks it $a$ holds (\textit{e.g.} $\ad{0\cdot 3\cdot 2\cdot 1\cdot 1} =2$). If $a\in \mathbb{N}^*$, the \textbf{collapse} of $a$, written $\ovl{a}$, is obtained by replacing in $a$ very track $>3$ by 2, \textit{e.g.} $\ovl{0 \cdot 5 \cdot 1 \cdot 3\cdot 2}=0\cdot 2 \cdot 1\cdot 2\cdot 2$.

A \textbf{tree} $A$ of $\bbN^*$ is a non-empty subset of $\bbN^*$ that is downward-closed for the prefix order ($a\leqslant a'\in A$ implies $a\in A$).

A subset $F\subset \bbN^*$ is a \textbf{forest} if $F=A-\{\epsi\}$ for some tree $A$ such that $0,\,1 \notin F$.


\subsection{Infinitary Terms}

\label{subsecTerms}


Let $\scrV$ be a countable set of term variables.
The set of terms $\Lambda^{111}$ is defined \textit{coinductively}:
$$t,\,u~ ::=~ x~ \|~ \lambda x.t~ \|~ tu $$

The parsing tree of $t\in \Lambda^{111}$, also written $t$, is the
labelled tree on $\Sigma_t:=\mathscr{V}\cup \{\lambda x~ |~ x\in
\mathscr{V}\}\cup \{\symbol{64}\}$ defined coinductively by: $\supp
(x)=\{\epsi\}$ and $x(\epsi)=x$, $\supp (\lambda x.t)= \{\epsi\}\cup
0\cdot \supp (t),~ (\lambda x.t)(\epsi)=\lambda x$ and $(\lambda
x.t)(0\cdot b)=t(b)$, $\supp (tu)=\{\epsi\} \cup 1\cdot \supp (t) \cup
2\cdot \supp (u)$, $(tu)(\epsi)=\symbol{64}$, $(tu)(1\cdot b)=t(b)$ and
$(tu)(2\cdot b)= u(b)$.


The abstraction $\lambda x$ binds $x$ in $t$ and $\alpha$-equivalence can be defined properly~\cite{DBLP:journals/tcs/KennawayKSV97}.

The relation $t \rewb{b} t'$ is defined by induction on $b\in 
\{0,\,1,\,2\}^*$: $(\lambda x.r)s\rewb{\epsi} r[s/x]$, $\lambda 
x.t\rewb{0\cdot b} \lambda x.t'$ if $t\rewb{b}t'$, $t_1t_2\rewb{1\cdot 
b} t'_1t_2$ if $t_1\rewb{b} t_1,~ t_1t_2 \rewb{2\cdot b} t_1t_2 '$ if 
$t_2\rewb{b} t'_2$. We define \textbf{$\beta$-reduction} by $\rew 
=\bigcup\limits_{b\in \{0,\,1,\,2\}^* } \rewb{b}$.

If $b=(b_i)_{i\in \mathbb{N}}$ is an \textit{infinite} sequence of
integers, we extend $\ad{b}$ as $|\{i \in \mathbb{N}~|~
b_i\geqslant 2\}|$ and we say that $b$ is an infinite branch of $t\in
\Lambda^{111}$ if, for all $n\in \mathbb{N},~ b_0\cdot b_1\cdot \ldots \cdot b_n\in \supp(t)$.  The calculus
$\Lambda^{001}$ is the set of terms $t\in \Lambda^{111}$ such that,
for every infinite branch $b$ of $t$, $\ad{b}=\infty$ . Thus, for
001-terms, infinity is allowed, provided we descend infinitely many
times in application arguments.


We define  a reduction sequence of length $\leqslant \omega$ of
$001$-terms $t_0\rewb{b_0}t_1\rewb{b_1} t_2\ldots $ to be \textbf{strongly
converging} if it is finite or if $\lim \ad{b_n}=+\infty$. See
\cite{DBLP:journals/tcs/KennawayKSV97} for an in-depth study of
\textbf{strongly converging reduction sequences (s.c.r.s)}. A compression
property allows us to consider only sequences of length $\leqslant
\omega$ without loss of generality. Assuming strong convergence, let
$b\in \mathbb{N}^*$ and $N\in \mathbb{N}$ s.t. $\forall n\geqslant N,~
\ad{b_n}> \ad{b}$. Then, either $\forall n\geqslant N,~ b\notin \supp
(t_n)$ or $\forall n\geqslant N,~ b\in \supp (t_n)$ and
$t_n(b)=t_N(b)$. Let $B'$ be the set of all the $b\in \mathbb{N}^*$ in
the latter case and $t'$ the labelled tree define by $\supp (t')=B'$
and $t'(b)=t_N(b)$ for any large enough $N$. We notice that $t'\in
\Lambda^{111}$. Actually, $t'$ is a 001-term (because at fixed
applicative depth, $t'$ must be identical to a $t_N$, for some large
enough $N$) and we call $t'$ the \textbf{limit} of the s.c.r.s. The
notation $t\rew ^\infty t'$ means that there is a s.c.r.s. starting
from $t$, whose limit is $t'$.

\subsection{Rigid Types}

If $X$ is a set, a \textbf{(partial) sequence} of $X$ is a familly $x=(x_k)_{k\in K}$ s.t. $K\subseteq \bbN-\{0,\,1\}$ and $x_k \in X$ for all $K$. We say $x_k$ is placed on \textbf{track} $k$ inside $(x_k)_{k\in K}$ and $K$ is the set of \textbf{roots} of $x$: we write $K=\Rt(x)$.

Let $\mathscr{X}$ be a countable set of types variables (metavariable $\alpha$).
The sets of (rigid) types $\Types^{+}$ (metavariables $T$, $T_i$, \ldots)
and rigid (sequence) types $\FTypes^{+}$ (metavariables $F$, \ldots)
are coinductively defined by:
$$ 
\begin{array}{lll} 
T & ::= &  \alpha ~ \|~ F \rightarrow T \\
F  & ::= &  (T_k)_{k\in K}
\end{array}$$

\begin{rmk}
  \begin{itemize}
  \item
  The \textbf{sequence type (seq.t.) $F=(T_k)_{k\in K}$} is a sequence of types in the above acception and is seen as an intersection of the types $T_k$ it holds.
\item If $U=F\rew T$, we set $\Tl(U)=F$ and $\Hd(U)=T$ (\textbf{tail} and \textbf{head}). 
  \end{itemize}
\end{rmk}

The equality between two types (resp. sequence types) is defined by mutual coinduction: 
$F\rew T=F'\rew T'$ if $F=F'$ and $T=T'$ and $(T_k)_{k\in K}=(T'_k)_{k \in K'}$ if $K^1=K^2$ and for all $k\in K,~ T_k=T_k'$.
It is a \textbf{syntactic equality} (unlike multiset equality). A $\ttS$-type can only be written one way.


The support of a type (resp. a sequence type), which is a tree of $\bbN^*$ (resp. a forest),  is defined by mutual coinduction: $\supp(\alpha)=\epsi,~ \supp(F\rew T)=\{\epsi\}\cup \supp(F)\cup 1\cdot \supp(T)$ and $\supp((T_k)_{k\in K})= \bigcup\limits_{k\in K} k\cdot \supp(T_k)$

A type of $\Types^+$ is in the set $\Types$ if its support does not hold an infinite branch ending by $1^{\omega}$. A sequence type $\FTypes^+$ is in $\FTypes$ if it holds only types of $\Types$. A (sequence) type is said to be finite when its support  is. We write $\eft$ for the forest type whose support is empty and 
$(T)_{i \in \{k\}}$ (only one type $T$, on track $k$) will simply be written $k\cdot T$.

When we quotient the sets $\Types$ and $\FTypes$ by a suitable congruence (collapsing the order in nested sequences), we get the set of types and multiset types of system $\scrM$ (\textit{cf.} Appendix \ref{appM}).

We say that a family of seq.t. $(F^i)_{i\in I}$ is \textbf{disjoint} if the $\Rt(F^i)$ ($i$ ranging over $i$) are pairwise disjoint. This means that there is no overlapping of typing information between the $F^i$. In that case, we define the \textbf{join} of $(F^i)_{i \in I}$ as the seq.t. $F$ s.t. $\Rt(F)=\bigcup\limits_{i\in I} \Rt(F^i)$ and, for all $k\in \Rt(F)$, $F_k=F_k^i$ where $i$ the unique index s.t. $k\in \Rt(F^i)$.

\subsection{Rigid Derivations}

\label{subsecDefRD}

A \textbf{(rigid) context} $C$ is a  function from $\mathscr{V}$
to $\FTypes$. The context $C-x$ is defined by $(C-x)(y)=C(y)$ for any
$y\neq x$ and $(C-x)(x)=\eft$.  We define the join of contexts
point-wise. A \textbf{judgment} is a sequent of the form $C \vdash t:\,T$, where $C$ is a context, $t$ a 001-term and $T\in \Types$.
The set $\Deriv$ of \textbf{(rigid) derivations} (denoted $P$) is defined
coinductively by the following rules:

\begin{center}
\begin{prooftree}
\Infer{0}[\ax]{x:\,k\cdot T \vdash x:\,T }
\end{prooftree}
\hspace{3cm}
\begin{prooftree}
\Hypo{C\vdash t:\, T \posPr{0}}
\Infer{1}[\abs]{C-x \vdash \lambda x.t:~
C(x)\rightarrow T}
\end{prooftree}\\[0.4cm]

\begin{prooftree}
\Hypo{C \vdash t:\, (S_k)_{k\in K}\rightarrow T }
\Hypo{D_k \vdash u:\, S'_k}
\Delims{ \left( }{ \right)_{k\in K'} }
\Infer{2}[\app]{C \cup \bigcup\limits_{k\in K} D_k \vdash t u:\, T}
\end{prooftree}
\end{center}

\begin{addc}
  \begin{itemize}
    \item In the $\app$-rule, the right part of the application is a sequence of judgments and we must have $(S_k)_{k\in K}=(S'_k)_{k\in K'}$ (syntactic equality).
\item Still in the $\app$-rule, the contexts must be disjoint, so that no track conflict occurs.
\end{itemize}
\end{addc}


In the axiom rule, $k$ is called an \textbf{axiom track}. In the $\app$-rule Once again, the judgment $(\Delta_k\vdash u:S_k)$ is called the \textbf{track $k$ premise} of the rule.

Once again, this
definition is very low-level, since the involved sequence types must be \textit{syntactically equal} to grant that the application is typable. The $\app$-rule may also be incorrect because two sequence types $C(x)$ or $D_k(x)$ (for a $x\in \mathscr{V}$) are not disjoint (\textbf{track conflict}). However, if we change wisely the values given to the axioms tracks, we can always assume that no conflict occurs for a specific axiom rule (for instance, using a bijection between $\bbN$ and a countable disjoint union of $\bbN$).

We can define \textbf{isomorphisms of derivations}. It is formally done in
Appendix \ref{appIso}. Concretely, $P_1$ and $P_2$ are isomorphic, written
$P_1\equiv P_2$, if they type the same term, there is well-behaved labelled tree isomorphism between their support and use isomorphic types and contexts. In that case, we can define type isomorphisms that are compatible in some
sense with the typing rules in the two derivations $P_1$ and $P_2$.

\subsection{Components of a Rigid Derivations and Quantitativity}
\label{subsecBisupp}

Thanks to rigidity, we can designate, identify and name every part of a derivation, thus allowing to formulate many associate, useful notions.

We can define the \textbf{support} of a derivation $P\rhd C\vdash t:T$: $\supp(P)=\epsi$ if $P$ is an axiom rule, $\supp(P)=\{\epsi\}\cup 0\cdot \supp(P_0)$ if $t=\lambda x.t_0$ and $P_0$ is the subderivation typing $t_0$, $\supp(P)=\{\epsi\}\cup 1\cdot \supp(P_1) \cup \bigcup\limits_{k \in K} \supp(P_k)$ if $t=t_1\,t_2,~ P_1$ is the left subderivation typing $t_1$ and $P_k$ the right subderivation proving the track $k$ premise. The $P_k$ ($k\in K$) are called \textbf{argument derivations}.


If $a\in \supp(P)$, then $a$ points to a judgment inside $P$ -- say this judgment is $C(a)\vdash t\rstr{\ovl{a}}:T(a)$: we say $a$ is a \textbf{position} of $P$. Now, let us locate ourselves at position $a$: if $c\in \supp(T(a))$, then $c$ is a pointer to a type symbol ($\alpha$ or $\rew$) in the type on the right side of the sequent nested a position $a$. We call the pair $(a,c)$ a \textbf{right biposition}: it points to a position in a type of a judgment nested in a judgment. Likewise, if $x\in \TermV$ and $k\cdot c\in C(a)(x)$, the pair $(c,x)$ points to a type symbol inside seq.t. $C(a)(x)$ (on the left side of the sequent) and we call the triple $(a,x,c)$ a \textbf{left biposition}. The \textbf{bisupport} of $P$, written $\bisupp(P)$ is the set of bipositions inside $P$ and if $\bip \in \bisupp(P)$, $P(\bip)$ is the nested type symbol that $\bip$ points at.


If $a\in A:=\supp (P)$ and $x\in \mathscr{V}$, we set $\Ax(a)(x)=\{ a' \in A~ | a'\geqslant a~ \text{and}~  t(\ovl{a'})=x\}$ if $x$ is free at pos. $\ovl{a}$ and  $\Ax(a)(x)=\emptyset$ if not: it is the set of (positions of) axiom leaves 
typing $x$ above $a$. If $a\in A$ is an axiom, we write $\tr(a)$ for its associated axiom track. 
The presence of an infinite branch inside a derivation makes it possible that a type in a context is not created in an axiom rule. This yields the notion of \textbf{quantitative derivation}, in which this does not happen:

\begin{defi}
A derivation $P$ is \textbf{quantitative} when, for all $a\in A$ and $x\in
\mathscr{V}$, $C(a)(x)=\bigcup\limits_{a'\in \Ax(a)(x)} \tr(a')\cdot T(a')$.
\end{defi}


Now, assume $P$ is quantitative. For all $a\in A$ and $x\in
\mathscr{V}$, we set $\AxTr(a)(x) = \Rt (C(a)(x))$. For all $a\in A,~
x\in \mathscr{V}$ and $k\in \AxTr(a)(x)$, we write $\pos(a,\,x,\, k)$
for the unique position $a'\in \Ax(a)(x)$ such that $\tr(a')=k$.


\section{Dynamics}




\label{secDynamics}

In this section, we explore the way reduction is performed inside a
derivation and introduce the notion of approximations and approximable
derivations. We assume $t|_b=(\lambda x.r)s$ and $t \rewb{b} t'$ and we consider a  derivation $P$ s.t. $P \rhd C \vdash t:\,T$. The letter $a$ will stand
for a representative of $b$ and the letter $\alpha$ for positions
inside $A:=\supp (P)$ (and \textit{not} for type variables).

\subsection{Residual (bi)positions}
\label{subsecResBip}

When $a\in \Rep_A(b)$, we set $\RedTr(a)=\Rt(C(a\cdot 10)(x))$. For $k\in \RedTr(a\cdot 10)(x)$, we write $a_k$ for the unique $a_k\in \mathbb{N}^*$ such that $\pos(a\cdot 10,\,x,\, k)=a\cdot 10\cdot a_k$. 

Assume $t\rstr{\ovl{a}}$ is a redex $(\lambda x.r)s$. We want to grant subject reduction according to the picture below:

\begin{center}
\begin{tikzpicture}

  
  \draw (0.3,3.6) node {\fnsz $P_r$} ; 
  \draw [dotted] (0.3,1.8)  -- ++(0,1.5) ;
  
  \draw (5.95,2.65) node{\fnsz $_{k\in K}$};
  \draw (0.6,2.97) node {\Big( } ;
  \draw (5.6,2.97) node {\Big) } ;
  \draw (5.1,3.07) node[right]{\fnsz $\ax$ };
  \draw (0.7,3.07) -- ++(4.4,0) ;
  \draw (0.7,2.8) node [right] {\fnsz $x:\, (S_k)_k \vdash x:\,S_k\posPr{a\cdot 10\cdot a_k}$  };

  \draw (-0.8,1.15) node [right] {\fnsz $C,\, x:\,(S_k)_{k\in K} \vdash r:\,T \posPr{a\cdot 10}$} ;  
\draw (-0.8,0.92) -- ++ (4.7,0);
\draw (-0.8,0.65) node [right] {\fnsz $C \vdash \lambda x.r:\,  (S_k)_{k\in K}\rew T\posPr{a\cdot 1}$} ;


\draw (7.65,1.92) node {\fnsz $_{k\in K}$};
\draw (4.3,1.25) node {\Large $\Bigg($} ;
\draw (7.45,1.25) node {\Large $\Bigg)$} ;
\draw (5.05,1.8) node {\fnsz $P_k$ };
\draw [dotted] (5.05,0.97) -- ++(0,0.6) ;
\draw (4.35,0.65) node [right]{\fnsz $D_k\vdash s:\, S_k \posPr{a\cdot k}$} ;

\draw (-0.8,0.42) -- ++(8.2,0) ;
  \draw (0,0) node [right] {\fnsz $C \cup \bigcup\limits_{k\in K} D_k \vdash (\lambda x.r)s:\, T\posPr{a}$} ;


\draw (8,0.77) node {$\leadsto$};


\draw (12.8,2.77) node {\fnsz $_{k\in K}$};
\draw (9.25,2.1) node {\Large $\Bigg($} ;
\draw (12.6,2.1) node {\Large $\Bigg)$} ;
\draw (10.1,2.65) node {\fnsz $P_k$ };
\draw [dotted] (10.1,1.8) -- ++(0,0.6) ;
\draw (9.4,1.5) node [right]{\fnsz $D_k\vdash s:\, S_k \posPr{a\cdot a_k}$} ;

  
  \draw (9,2.3) node {\fnsz $P_r$} ; 
  \draw [dotted] (9,0.5)  -- ++(0,1.5) ;
  
\draw (8.2,0) node [right] {\fnsz  $C \cup \bigcup\limits_{k\in K} D_k \vdash  r[s/x]:\, T\posPr{a}$};


\draw (-0.8,-0.6) node[below right] { \parbox{7.8cm}{\textbf{Derivation typing $(\lambda x.r)s$ }
{\small    \begin{itemize}
  \item $P_r$ is the subderiv. (above $a$) typing $r$.
  \item In $P_r$, the axiom rule (typing $x$) using track $k$ is at position $a\cdot 10\cdot a_k$.
      \item The arg. deriv. $P_k$ yields the track $k$ premise.
      \end{itemize}}
} } ; 


\draw (8.2,-0.6) node[below right] { \parbox{5cm}{\textbf{Derivation typing $r[s/x]$ }
    {\small 
    \begin{itemize}
    \item The application and abstraction of the redex have been destroyed.
    \item In $P_r$, the arg. deriv. $P_k$ has replaced $x$-axiom using track $k$.
    \end{itemize}}
} } ; 

\end{tikzpicture}
\end{center}


Notice how this transformation is deterministic: for instance, assume $7 \in K$. There must be an axiom rule typing $x$ using axiom track 7 \textit{e.g.} $x:\, 7\cdot S_7\vdash x:\,S_7$ at position $a\cdot 10\cdot a_7$ and also a subderivation at argument track 7, namely, $P_7$
 concluded by $s:S_7$
at position $a\cdot 7$. Then, when we fire the redex at position $b$, the subderivation $P_7$ \textit{must} replace the axiom rule on track 7, even if there may be several $k\neq 7$ such that $S_k=S_7$ (compare with \S~\ref{subsecNonD}).



The figure above represents the quantitative case but the following construction does not assume $P$ to be quantitative (although motivated by it). The notion of \textit{residual} (right bi)positions tells us where a (right bi)position inside $P$ will be placed in derivation $P'$. Assume $\alpha \in A,~ \ovl{\alpha}\neq a,\, a\cdot 1,\, a \cdot 10\cdot a_k$ for no $a\in \Rep_A(b)$ and $k\in \RedTr(a)$. The \textbf{residual position} of $\alpha$, written $\Res_b(\alpha)$, is defined as follows \textit{i.e.} (1)  if $\alpha \geqslant a\cdot k\cdot \alpha_0$ for some $a\in \Rep(b)$ and $k\geqslant 2$, then $\Res_b(\alpha)=a\cdot a_k \cdot \alpha_0$ (2) if $\alpha =a\cdot 10 \cdot \alpha_0$ for some $a\in \Rep(b)$, then  $\Res_b(\alpha)=a\cdot \alpha_0$ and (3) if $\ovl{a}\ngeqslant b$, $\Res_b(\alpha)=a$. 


We set $A'= \codom (\Res_b)$ (\textbf{residual support}). Now, whenever
$\alpha':= \Res_b(\alpha)$ is defined, the \textbf{residual
biposition} of $\bip:=(\al,\gam)\in \bisupp(P)$ is $\Res_b(\bip)=(\alpha',\,\gamma)$. We
notice that $\Res_b$ is an \textit{injective}, \textit{partial}
function, both for positions and right bipositions. In particular,
$\Res_b$ is a bijection from $\dom (\Res_b)$ to $A'$ and we write
$\Res_b^{-1}$ for its inverse. For any $\alpha'\in A'$, let
$C'(\alpha')$ be the context defined by $C'(\alpha')=(C(\alpha)-x) \cup
\bigcup\limits_{k\in K(\alpha)} C(\alpha\cdot k)$, where
$\alpha=\Res_b^{-1} (\alpha')$ and $K(a)=\Rt(C(a)(x))$.
Notice that $C'(\alpha)=C(\alpha)$ for any $\alpha \in A$ s.t. $\ovl{\alpha} \ngtr b$, \textit{e.g.} $C'(\epsi)= C(\epsi)$.

\subsection{Deterministic subject reduction and expansion}
\label{subsecSRSE}

Let $P'$ be the labelled tree such that $\supp (P')=A'$ and $P'(\alpha')$ is $C'(\alpha') \vdash t'|_{\alpha'}:\, T(\alpha)$ with $\alpha'=\Res_b(\alpha)$. We claim that $P'$ is a correct derivation concluded by $C\vdash t':T$: indeed, $\ovl{A'}\subset \supp (t')$ stems from $\ovl{A}\subset \supp (t)$. Next, for any $\alpha'$ and $\alpha=\Res_b^{-1}(\alpha)$, 
$t'(\ovl{\alpha'})=t(\ovl{\alpha})$ and the rule at position $\alpha'$
is correct in $P'$ because the rule at position $\alpha$ in $P$ is 
correct (for the abstraction case, we notice that $t'(\ovl{\alpha'})=\lambda y$ implies $C'(\alpha')(y)=C(\alpha)(y)$).


\begin{prop}[Subject Reduction]
If $t\rewb{b} t'$ and $C\vdash t:\,T$ is derivable, then so is $C\vdash t':\,T$.
 \end{prop}


With the above notations, we also write $P\rewb{b} P'$.
The subject-expansion property hold for quantitative
derivations. Namely, we build a derivation $P\rhd C\vdash t:\,T$ from a derivation $P'\rhd C\vdash t':\,T$, so that $P\rewb{b} P'$ by using a converse method. There are several possibilities to build such a $P$, because we have to choose
an axiom track $k$ for each occurrence of $x$ inside $P$ (in that case, $x$ is \textit{quantitatively} typed). For
instance, we can fix an injection $\code{\cdot}$ from $\mathbb{N}^*$
to $\mathbb{N}-\{0,\,1\}$ and to choose the track $\code{\alpha}$ for
any axiom rule created at position $\alpha$.

\begin{prop}[Subject Expansion]
If $t\rewb{b} t'$ and $C\vdash t':\,T$ is derivable, then so is $C\vdash t:\,T$.
\end{prop}


Determinism make subject reduction/expansion well-behaved w.r.t. isomoprhism:

\begin{lemma}
\begin{itemize}
\item If $P_1\equiv P_2$, $P_1\rewb{b} P_1'$ and $P_2\rewb{b} P_2'$, then $P_1'\equiv P_1'$.
\item Assume $P_1$ and $P_2$ quantitative: if $P_1 \rewb{b} P',~ P_2\rewb{b} P'$, then $P_1\equiv P_2$.
\end{itemize}
\end{lemma}

\section{Approximable Derivations and Unforgetfulness}

\label{secApproxUF}

\subsection{Approximability and Monotonicity}

\label{subsecApprox}

We define here our validity condition \textit{i.e.} approximability.
Morally, a derivation is approximable if all its bipositions are 
relevant.


\begin{defi}
\begin{itemize}
\item  Let $P$ and $P_*$ be two derivations. We say 
$P_*$ is an \textbf{approximation} of $P$, and we write $P_* \leqfty P$, if $\bisupp (P_*) \subset \bisupp P$ and for all $\bip \in \bisupp (P_*),~ P_*(\bip)=P(\bip)$.
\item We write $\Approx_{\infty}(P)$ for the set of approximations of a derivation $P$ and $\Approx(P)$ for the set of {\it finite} approximations of $P$.
\end{itemize}
\end{defi}

Thus, $P_* \leqfty P$ if $P_*$ is a {\it sound} restriction of $P$ of a subset of $\bisupp(P)$. We usually usually write $\supf P$ for a {\it finite} approximation of $P$ and in that case only, write $\supf P \leqslant P$ instead of $\supf P\leqfty P$. Actually, $\leqfty$ and $\leqslant$ are associated to lattice structures induced by the set-theoretic union and intersection on bisupports :

\begin{theorem}
 The set of derivations typing a same term $t$ endowed with $\leqfty$ is a directed complete partial order (dcpo) .
  \begin{itemize}
  \item If $D$ is a directed set of derivations typing $t$:
\begin{itemize}
  \item The {\bf join} $\vee D$ of $D$ is the labelled tree $P$ defined by $\bisupp(P)=\cup_{P_*\in D} \bisupp(P_*)$ and $P(\bip) = P_*(\bip)$ (for any $P_*\in D$ s.t. $\bip \in \bisupp (P_*)$), which is a derivation.
  \item  The {\bf meet} $\wedge D$ of $D$ is the labelled tree $P$ defined by $\bisupp(P)=\cap_{P_*\in D} \bisupp(P_*)$ and $P(\bip) = P_*(\bip)$, which also is a derivation.
\end{itemize}
  \item If $P$ is a derivation typing $t$, $\Approx_\infty(P)$ is a complete lattice and $\Approx(P)$ is a finite lattice.
  \end{itemize}
\end{theorem}

This allows us to define now the notion of approximability, related to the {\it finite approximations} :

\begin{lemma} \label{lemRedBij}
\begin{itemize}
\item Reduction is monotonic: if $\fP \leqslant P,~ \fP \rewb{b} \fP'$ and $P\rewb{b}P'$, then $\fP' \leqslant P'$.
\item Moreover, if $P\rewb{b} P'$, for any $\supf P'\leqslant P$, 
there is a unique $\supf P\leqslant P$ s.t. $\supf P \rewb{b} P'$.
\end{itemize}
\end{lemma}

\begin{defi}
A derivation $P$ is \textbf{approximable} if, for all finite $\supo B\subset \bisupp (P)$, there is a $\supf P\leqslant P$ s.t. $\supo B\subset \bisupp (\supf P)$.
\end{defi}

\begin{lemma} \label{lemApproxRed}
\begin{itemize}
\item If $P$ is not quantitative, then $P$ is not approximable.
\item If $P$ is quantitative and $P\rewb{b} P'$, 
then $P$ is approximable iff $P'$ is approximable.
\end{itemize}
\end{lemma}

\begin{proof}
\begin{itemize}
\item Assume $\bip=(a,\,x,\, k)\in \bisupp(P)$ is such that there is no $a_0\in \Ax(a)(x)$ s.t. $\tr(a_0)=k$. No finite $\supf P\leqslant $ could contain $\bip$, because, by typing constraints, it would also contain all the $(a',\,x,\,k) \in \bisupp(P)$ (there must be infinitely many).
\item Assume $P$ approximable. Let $\supo B'\subset \bisupp (P')$ be finite.
We set $\supo B=\Res_b^{-1}(\supo B')$. There is $\supf P\leqslant P$ s.t. 
$B\subset \bisupp (\supf P)$. Let $\supf P'$ be the reduced of $\supf P$ 
at position $b$. Then $\supo B'\subset \bisupp (\supf P')$. \\
The proof is the same for the converse. However, $\Res_b$ and 
$\Res_b^{-1}$ are not defined for every biposition (\textit{e.g.}
left ones) and our argument is faulty. It is not hard to avoid this
problem (it is done in Appendix \ref{appEquinecessity}), using 
\textit{interdependencies} between bipositions.
\end{itemize} 
\end{proof}

\subsection{Unforgetfulness}

\label{subsecUF}

The left side of an arrow is regarded as having a negative
\textbf{polarity} and its right side as having a positive polarity.
The type characterization of the WN terms in the finitary calculus
(th. 4, ch.3, \cite{DBLP:books/daglib/0071545}) relies on the notion
of positive and negative occurrences of the ``meaningless'' type
$\Omega$. We adapt this criterion, which motivates:

\begin{defi} \label{defiUF}
\begin{itemize}
\item The sets $\EFO^+(U)$ and $\EFO^-(U)$ are defined by mutual coinduction 
for $U$, a type $T$ or a forest type $F$. $\EFO$ stands for 
\textbf{empty forest occurrences} and $\pm$ indicates their polarity.
The symbol $\mp$ is $-/+$  when $\pm$ is $+/-$. 
\begin{itemize}
\item $\EFO^{\pm}(\alpha)=\emptyset$ for $\alpha \in \mathscr{X}$.
\item $\EFO^{\pm}(\eft \rew T)=\{\epsi\}\cup 1\cdot \EFO^{\pm}(T)$
\item If $F\neq \eft$, $\EFO^{\pm}(F\rew T)=\EFO^{\mp}(F)\cup 1\cdot \EFO^{\pm}(T)$.
\item $\EFO^{\pm}(F)=\bigcup\limits_{k\in K} 
k\cdot \EFO^{\pm}(T_k)$ with $F=(T_k)_{k\in K}$.
\end{itemize}
\item We say a judgment $C\vdash t:\,T$ is \textbf{unforgetful}, when, for all $x\in \mathscr{V}$, $\EFO^-(C(x))=\emptyset$ and $\EFO^+(T)=\emptyset$. A derivation proving such a judgment is also said to be unforgetful.
\end{itemize}
\end{defi}

\begin{lemma}
If $P\rhd C \vdash t:\,T$ is an unforgetful derivation typing a HNF 
$t=\lambda x_1\ldots x_p.(x\,t_1)\ldots t_q$, then, there are 
unforgetful subderivations of $P$ typing $t_1$, $t_2$,..., $t_q$.\\
Moreover, if $P$ is approximable, so are they.
\end{lemma}

\begin{proof}
Whether $x=x_i$ for some $i$ or not, the unforgetfulness condition grants that every argument of the head variable $x$ is typed, since $\eft$ cannot occur negatively in its unique given type.
\end{proof}

\begin{lemma} \label{lemTypHN}
If $P \rhd C\vdash t:\, T$ is a \textit{finite} derivation, then 
$t$ is head normalizable. 
\end{lemma}

\begin{proof}
\begin{itemize} 
\item By typing constraints, the head redex (if $t$ is not already in 
HNF) must be typed. 
\item When we reduce a typed redex, the number of rules of the derivation must
strictly decrease (at least one $\symbol{64}$-rule and one $\lambda x$-rule
disappear). See Appendix \ref{appQuantArg}.
\item Thus, the head-reduction strategy must halt at some point.
\end{itemize}
\end{proof}

\begin{prop} \label{propNDTypWN}
If a term $t$ is typable by a unforgetful approximable derivation,
then it is WN (in other words, it is HHN).
\end{prop}

\begin{proof}
Consequence of the two former lemmas.
\end{proof}

\subsection{The infinitary subject reduction property}
\label{secRed}


In this section, we show how to define a derivation $P'$ typing $t'$ 
from a derivation typing a term $t$ that strongly converges towards $t'$. Actually, when a reduction is performed at depth $n$, the contexts and types are not affected below $n$. Thus, a s.c.r.s. stabilizes contexts and types at any fixed depth. It allows to define a derivation typing the limit $t'$.

The following \textbf{subject substitution lemmas} are very useful while
working with s.c.r.s.:

\begin{lem}
Assume $P\rhd C \vdash t:\, T$ and for all $a\in A:=\supp (P),~ t(\ovl{a})
=t'(\ovl{a'})$ (no approximability condition).\\
Let $P'$ be the derivation obtained from $P$ by substituting $t$ with $t'$
\textit{i.e.} $\supp (P')=\supp (P)$ and for all $a\in A,~ P'(a)=C(a)
\vdash t' |_{\ovl{a}}: \,T(a)$.\\
Then $P'$ is a correct derivation.
\end{lem}

\begin{lem}
Assume $\supf  P\leqslant P,~ P\rhd C \vdash t:\,T, P'\rhd C'\vdash t':\,T'$ 
and for all $a\in \supf A:=\supp ( \supf P),~ t(\ovl{a})=t'(\ovl{a}),~ 
C(a)=C'(a)$ and $T(a)=T'(a)$.\\
Let $\supf P'$ be the derivation obtained by replacing $t$ with $t'$.
Then $\supf P'\leqslant P'$.
\end{lem}


Now, assume that:
\begin{itemize}
\item $t\rightarrow^{\infty} t'$ is a s.c.r.s. Say that this sequence is $t=t_0 \rewb{b_0} t_1\rewb{b_1} ... \rewb{b_{n-1}} t_n \rewb{b_n}t_{n+1}\rewb{b_{n+1}} \ldots $ with $b_n\in \{0,~ 1,~ 2\}^*$ and $\ad{b_n}\longrightarrow +\infty$.
\item There is a quantitative derivation $P \rhd C \vdash t:\,T$ and  $A=\supp(P)$.\\[-0.15cm]
\end{itemize}

By performing step by step the s.c.r.s. $b_0,~ b_1,\ldots$, we get a sequence of derivations $P_n\rhd \Gamma_n \vdash t_n:\, T_n$ of support $A_n$ (satisfying $C_n(\epsi)=C(\epsi)$ and  $T_n(\epsi)=T(\epsi)$). When performing $t_n\stackrel{b_n}{\rightarrow}t_{n+1}$, notice that $C_n(c)$ and $T_n(c)$ are not modified for any $c$ such that $b_n\nleqslant \ovl{c}$.

Let $a\in \mathbb{N}^*$ and $N\in \mathbb{N}$ be such that, for all
$n\geqslant N,~ |b_n|>|a|$. There are two cases:
\begin{itemize}
\item $a\in A_n$ for all $n\geqslant N$. Moreover,
  $C_n(a)=C_N(a),~ T_n(a)=T_N(a)$ for all $n\geqslant N$, and $t'(\ovl{a}=t_n(\ovl{a})=t_N(\ovl{a})$.  
\item $a\notin A_n$ for all $n\geqslant N$.
\end{itemize}

We set $ A'=\{a\in \mathbb{N}^*~ | ~ \exists N,~ \forall n \geqslant N, ~
a\in A_n\}$. We define a labelled tree $P'$ whose support is $A'$ by
$P'(a)=C_n(a)\vdash t'|_{\ovl{a}}:\,T_n(a) $ for any $n\geqslant N(|a|)$ (where $N(\ell)$ is the smallest rank $N$ such that $\forall n \geqslant N,~ |a_n|>\ell$).

\begin{prop}
The labelled tree $P'$ is a derivation (the subject-reduction property holds for s.c.r.s. without considering approximability).
\end{prop}

\begin{proof} 
Let $a\in A'$ and $n\geqslant N(|a|+1)$. Thus, $t'(a)=t_n(a)$ and the 
types and contexts involved at node $a$ and its premises are the same in 
$P'$ and $P_n$. So the node $a$ of $P'$ is correct, because it is 
correct for $P_n$. 
\end{proof}

\begin{prop}
If $P$ is approximable, so is $P'$.
\end{prop}

\begin{proof}
  Assume $\supf P\leqslant P$. Let $N=|\bisupp(\supf P)|$. Notice that $\supf P$ cannot type any position whose length is greater than $N$.

Then, $t$ can be reduced (in a finite number $\ell$ of steps) into
a term $t_\ell$ such that $t_\ell(b)=t'(b)$ for all $b\in \{0,~ 1,~ 2\}^*$
such that $|b|\leqslant N$.

We have $\supf P_\ell \leqslant P_\ell$ (monotonicity). Let $\supf P'$ be the
derivation obtained by replacing $t_\ell$ by $t'$ in $\supf P_\ell$.
The Substitution Lemmas entail that $\supf P'$ is a correct derivation
and $\supf P'\leqslant P'$.
\end{proof}

\section{Typing Normal Forms and Subject Expansion}

\label{secNF}

In this section, we characterize all the possible quantitative derivations typing a normal form $t$, and show all of them to be approximable.

\subsection{Positions in a Normal Form}

We write $a\prec a'$ when there is $a_0$ such that $a_0\leqslant a,~ a_0\leqslant a'$, $\ad{a}=\ad{b}$ and $\ad{a'}\geqslant \ad{b}$. The relation $a\prec a'$ is a \textit{preorder} and represents an "applicative priority" w.r.t. typing.
Namely, assume $\ovl{a},~ \ovl{a'}$ are in $\supp (t)$, $a\prec a'$ and $P$ types $t$. Then, $a'\in \supp(P)$ implies $a \in \supp (P)$. For instance, if $021037\in \supp(P)$, then $t(\ovl{02103})=\arob$ and $02031$, which is this application left-hand side, should also be in $\supp(P)$, as well as every prefix of $021037$.

This motivates to say that a $A\subset \mathbb{N}^*$ is a \textbf{derivation support (d-support)} of $t$ if $\ovl{A}\subset \supp (t)$ and $A$ is downward closed for $\prec$.  We will show that, in that case, there is actually a $P$ typing $t$ s.t. $A=\supp(P)$ (this holds only because $t$ is a NF).\\[-0.15cm]

If $t$ is a NF and $\ovl{a}\in \supp (t)$ ($a$ is a d-position in $t$), we may have:
\begin{itemize}
\item $t|_{\ovl{a}}=\lambda x_1\ldots x_n.u$ where $u$ is not an abstraction. We set then $\mathring{a}=a\cdot 0^n$. When $n\neq 0$, we say that $a$ is an \textbf{abstraction position}.
\item $t_{\ovl{a}}$ is not an abstraction. Then there is a smallest prefix $a'$ of $a$ such that $a=a'\cdot 1^n$. We set then $\mathring{a}=a'$. When $n\neq 0$, we say $a$ is a \textbf{partial position}.\\[-0.15cm]
\end{itemize}
When  $a=\mathring{a}$, $a$ is called a \textbf{full position}. The set of full positions inside a d-support $A$ of $t$ is also written $\mathring{A}$. In the 3 cases, $\rdeg(a):=||a|-|\mathring{a}||$ is the \textbf{relative degree} of position $a$.

 We identify graphically $\mathring{a}$ and its collapse. 
 We have 3 kinds of positions : a position is full when we can choose freely the type it makes appear. 
 An abstraction position is a position that prefixes a full position by means of a sequence of abstraction and a partial position is a position that postfixes a full position by ....

\subsection{Building a Derivation typing a Normal Form}

\label{subsecSpCons}

We build here, from any given d-support $A$ of $T$ and function $T$ from 
$\mathring{A}$ to $\Types$, a quantitative derivation $P$ such that 
$\supp (P)=A$, giving the type $T(a)$ to $t|_{\ovl{a}}$ for any 
$a\in \mathring{A}$. Disregarding indexation problems, we must have,
by typing constraints: 
\begin{itemize}
\item If $a$ is an abstraction position - say $\rdeg(a)=n$ and $t|_{\ovl{
a}}= \lambda x_1\ldots x_n.t_{\mathring{a}}$ -, then 
$T(a)=C(a\cdot 0)(x_1)\rew \ldots
\rew C(a\cdot 0^n)(x_n) \rew T(\mathring{a})$, where $C(a)(x)$
is a s.t. containing every type given in $\Ax(a)(x)$.
\item If $a$ is partial - say $a=\mathring{a}\cdot
  1^n$ and $t|_{\mathring{a}}=t_{\ovl{a}}t_1\ldots t_n$ -, then
  $T(a)=\Rft_i (a) \rew \ldots \rew \Rft_n(a) \rew T(\mathring{a})$,
  where $\Rft_i(a)$ is the s.t. holding all the types given to $t_i$
  (below $\mathring{a}$). \\[-0.15cm]
\end{itemize}

If $(T_i)_{i\in I}$ is a family of types and $(k_i)_{\in I}$ a family of pairwise distinct integers $\geqslant 2$, the notation $(k_i\cdot T_i)_{i\in I}$ will denote the forest type $F$ s.t. $\Rt (F)=\{k_i ~ | ~ i\in I\}$ and $F\rstr{k}=T_i$ where $i$ is the unique index s.t. $k=k_i$.

We consider from now on an injection $a\mapsto  \code{a} $ from $\mathbb{N}^*$ to $\mathbb{N}-\{0,~ 1\}$. To each $a\in \mathbb{N}^*$, we attribute a fresh type variable $X_a$.

When $a$ is partial and $\rdeg(a)=n$ (and thus, $a=\mathring{a}\cdot 1^n$), we set, for $1\leqslant k$, $AP_k(a)=\{\mathring{a} \cdot 1^{n-k}\cdot \ell \in A~ |~ \ell \geqslant 2\}$ (AP stands for "argument positions").\\[-0.15cm]
\begin{itemize}
\item If $a\in A$ and $x\in \mathscr{V}$ is free at position $a$, we define the forest type $E(a)(x)$ by $E(a)(x)= (\code{a'}\cdot X_{a'})_{a'\in \Ax(a)(x)}$ ($\Ax(a)(x)$ is defined w.r.t. $A$). If $a\in A$ is partial and $a=\mathring{a}\cdot 1^n$ and $1\leqslant k \leqslant n$, we define the forest type $F_k(a)$ by $F_k(a)=(\code{a'}\cdot X_{a'})_{a'\in AP_k(a)}$. If $a\in A$ is full, we set $S(a)=T(a)$ (in that case, $S(a)$ does not hold any $X_k$).
\item If $a\in A$ is an abstraction position -- say $t\rstr{\ovl{a}}=\lambda x_1\ldots x_n.u$ where $t_{\mathring{a}}=u$), we set $S(a)=E(a\cdot 0)(x_1)\rew \ldots \rew E(\mathring{a})(x_n)\rightarrow T(\mathring{a})$.  If $a\in A$ is partial, we set $S(a)=F_1(a)\rightarrow F_2(a)\rightarrow \ldots \rightarrow F_n(a) \rightarrow T(\mathring{a})$. We we extend $T$ (defined on full positions) to $A$ by the following coinductive definition: for all $a\in A,~ T(a)=S(a)[T(a')/X_{a'}]_{a'\in \mathbb{N}^*}$. For all $a\in A$, we define the contexts $C(a)$ by $C(a)(x)=E(a)(x)[T(a')/X_{\code{a'}}]_{a'\in \Ax(a)(x)}$.
\end{itemize}

Those definitions are well-founded, because whether $a$ is $\lambda x$-position or a partial one, every occurrence of an $X_k$ is at depth $\geqslant 1$ and the coinduction is \textit{productive}. Eventually, let $P$ be the labelled tree whose support is $A$ and s.t., for $a\in A$, $P(a)$ is $C(a)\vdash t|_{\ovl{a}}:\, T(a)$.


\begin{prop}
The labelled tree $P$ is a derivation proving $C(\epsi)\vdash t:\,T(\epsi)$.
\end{prop}


\begin{proof}
  Let $a\in A$. Whether $t(\ovl{a})$ is $x$, $\lambda{x}$ or $\symbol{64}$, we check the associated rule has been correctly applied.  Roughly, this comes from the fact that the variable $X_{a'}$ is "on the good track" (\textit{i.e.} $\code{a'}$), as well as in $F_i(a)$, thus allowing to retrieve correct typing rules.
\end{proof}

\begin{defi}
The above method of building of a derivation $P$ typing a normal form $t$, from a d-support $A$ of $t$ and a function $T$ from full positions of $A$ to $\Types$, will be referred as the \textbf{trivial construction}.
\end{defi}

\begin{prop} \label{propNFhasNDD} 
A normal form $t\in \Lambda^{001}$ admits an unforgetful derivation.
\end{prop}

\begin{proof}
We set $A=\supp (t)$ and $T(a)=\alpha$ for each full position (where
$\alpha$ is a type variable). In that case, the trivial construction yields an unforgetful derivation of $t$.
\end{proof}

It is easy to check that the above derivations yield representatives
of every possible quantitative derivation of a NF (notice there is no
approximability condition):

\begin{prop} \label{propCharQTofNF} 
If $P$ is a quantitative derivation typing $t$, then the trivial
construction w.r.t. $A:=\supp (P)$ and the restriction of $T$ on $\mathring{A}$
yields $P$ itself.
\end{prop}

\subsection{Approximability} 

We explain here why every quantitative derivation $P$ typing a normal 
form is approximable (see Appendix \ref{appNF} for a complete proof). This means that we can build a finite derivation $\supf P \leqslant P$ containing any finite part $\supo B$ of $\bisupp (P)$. We will proceed by: 
\begin{itemize} 
\item Choosing a finite d-support $\supf A$ of $A$ \textit{i.e.} we will 
discard all positions in $A$ but finitely many.
\item Choosing, for each $T(a)$ s.t. $a$ is full, a finite part of 
$\supf T(a)$ of $T(a)$. 
\end{itemize} 
The trivial construction using $\supf A$ and $\supf T$ will yield a derivation 
$\supf P\leqslant P$ typing $t$. 

Namely, we fix an integer $n$ and discard every position $a\in A$ such that $\ad{a}>n$ or $a$ holds a track $\geqslant n$. It yields a finite d-support $A_n\subset A$. Then, for each $\mathring{a}$, we discard every inner position inside $T(\mathring{a})$ according to the same criterion. It yields finite types $T_n(\mathring{a})$. The trivial construction starting from $A_n$ and $T_n$ yields a (finite) derivation $P_n\leqslant P$. We prove then that $n$ can always be chosen big enough to ensure that $\supo B \subset \bisupp(P_n)$.

\subsection{The Infinitary Subject Expansion Property}

\label{secExpans}


In Section \ref{secRed}, we defined the derivation $P'$ resulting
from a s.c.r.s. from any (approximable or not) derivation
$P$. Things do not work so smoothly for subject expansion when we try
to define a good derivation $P$ which results from a derivation $P'$
typing the limit of a s.c.r.s. Indeed, approximability play a central
role w.r.t. expansion. Assume that:
\begin{itemize} 
\item $t\rightarrow^{\infty} t'$. Say through the s.c.r.s. $t=t_0 
\rewb{b_0} t_1\rewb{b_1} \ldots t_n \rewb{b_n} t_{n+1} \rew...$ with 
$b_n\in \{0,~ 1,~ 2\}^*$ and $\ad{b_n}\longrightarrow +\infty$.
\item $P'$ is an approximable derivation of $C'\vdash t':\,T'$.
\end{itemize} 

The main point is to understand how subject expansion works with a finite derivation $\supf P'\leqslant P'$. The technique of \S~ \ref{subsecExamples} can now be formally performed. 
Mainly, since $\supf P'$ is finite, for a large enough $n$, $t'$ can be replaced by $t_n$ inside $\supf P'$, due to the subject substitution lemmas \S~ \ref{secRed}, which yields a finite derivation $\supf P_n$ typing $t_n$. But when $t_n$ is typed instead of $t'$, we can perform $n$ steps of expansion (starting from $\supf P_n$) to obtain a finite derivation $\supf P$ typing $t$. Then, we define $P$ as the join of the $\supf P$ when  $\supf P'$ ranges over $\Approx(P')$. Complete proofs and details are to be found in Appendix \ref{appExpans}

\begin{prop}
The subject expansion property holds for approximable derivations and
strongly convergent sequences of reductions.
\end{prop}

Since subject-reduction and expansion of infinite length (in s.c.r.s) preserve unforgetful derivation, it yields our main characterization theorem :

\begin{theorem}
A term $t$ is weakly-normalizing in $\Lambda^{001}$ if and only if $t$ is typable by means of an approximable unforgetful derivation.
\end{theorem}

\begin{proof}
The $\Leftarrow$ implication 
is given by proposition \ref{propNDTypWN}. For the direct one:
assume $t$ to be WN and the considered s.c.r.s. to yield the NF $t'$ of $t$. 
Let $P'$ be an unforgetful derivation typing $t'$ (granted by proposition 
\ref{propNFhasNDD}). Then, the derivation $P$ obtained by the proposition below
is unforgetful and types $t$.
\end{proof}

\section{Conclusion}

We have provided an intersection type system characterizing
weak-normalizability in the infinitary calculus $\Lambda^{001}$.  The
use of functions from the set of integers to the set of types to
represent intersection -- instead of multisets or conjunctions --
allows to express a validity condition that could only be suggested in
De Carvalho's type assignment system.  Our type system system is
relatively simple and offers many ways to describe proofs
(\textit{e.g.} tracking, residuals).

It is then natural to seek out whether this kind of framework could be
adapted to other infinitary calculi and if we could also characterize
strong normalization in $\Lambda^{001}$, using for instance a memory
operator~\cite{BucciarelliKV15}. Although our derivations are very low-level objects, it can be shown they allow to represent any infinitary derivation of system $\mathscr{M}$~\cite{vialXX}. We would like to find alternatives to
the approximability condition, \textit{e.g.} formulating it only in term of
tracks. It is to be noticed that derivation approximations provide
\textit{affine} approximations that behave \textit{linearly} in Mazza's polyadic
calculus~\cite{DBLP:conf/lics/Mazza12}. Last, the type system
presented here can very easily adapted to the terms whose B\"ohm tree
may contain $\bot$, with the same properties. In particular, two terms
having the same B\"ohm tree can be assigned the same types in the same
contexts. We would like to investigate ways to get (partial forms of)
the converse.

\label{secConclusion}

\bibliographystyle{plain}
\bibliography{BiblioShort}

\newpage

\tableofcontents

\newpage

\newpage

\appendix

\section{Performing reduction in a derivation}

\label{appSR}

We assume here that $t|_{b}=(\lambda x.r)s$
$t\rewb{b} t'$.

\subsection{In De Carvalho's system $\mathscr{M}_0$}

We assume assume $\Pi \rhd \Gamma_\epsi \vdash t:\, \tau_\epsi$.
We consider the  (non-deterministic) transformation below, that we perform
on any subderivation of $\Pi$ corresponding to the position $b$.
For such a subderivation, it is meant that the $\Pi$ are the
argument subderivations typing $s$ and that
$\Pi_0$ is the subderivation typing $r$. We have indicated between
brackets the positions of the axiom leaves typing $x$, the bound
variable to be substituted during the reduction.

\begin{tikzpicture}
\draw (1.4,3.5) node {$\Pi_0$} ;

\draw (1.7,3.7) node [below right]{
	\begin{prooftree}
	\Infer{0}[ax]{x:\, [\sigma_i] \vdash x:\,\sigma_i}
	\Delims{\left( }{ \right)_{i\in I} }
	\end{prooftree}
} ;

\draw [dotted] (1.4,2.2) -- (1.4,3.2) ;
	\draw (0,2) node [below right] {$\Gamma,\, x:\,[\sigma_i]_{i\in I}
	\vdash r:\,\tau $} ;
\draw (0,1.5)--(3.4,1.5) ;
	\draw (0,1.5) node [below right] {$\Gamma \vdash \lambda x.r:\, 
	[\sigma]_{i\in I}\rew \tau $} ;

\draw (6.5,2.3) node [below] { 
	\begin{prooftree}
	\Hypo{\Pi_i }
	\Ellipsis{}{\Delta_i\vdash s:\, \sigma_i }
	\Delims{\left( }{ \right)_{i\in I} }
	\end{prooftree}
} ;

\draw (0,1)--(7.6,1) ;
	\draw (0,1) node [below right] {$\Gamma + \sum\limits_{i\in I} \Delta_i
	\vdash (\lambda x.r)s:\, \tau $} ;

\end{tikzpicture}
$\rightsquigarrow$
\begin{tikzpicture}
\draw (3.3,3.3) node [below] {
        \begin{prooftree}
        \Hypo{\Pi_i }
        \Ellipsis{}{\Delta_i\vdash s:\, \sigma_i }
        \Delims{\left( }{ \right)_{i\in I} }
        \end{prooftree} 
} ;

\draw (1.4,2.5) node {$\Pi_0$} ;
\draw [dotted] (1.4,1.2) -- (1.4,2.2) ;
        \draw (0,1) node [below right] {$\Gamma + \sum\limits_{i\in I} \Delta_i
        \vdash r[s/x]:\, \tau $} ;
\end{tikzpicture}

Thus, from the type-theoretical point of view, in De Carvalho's type
assignment system, reduction consists in:
\begin{itemize}
\item Destroying the application and
abstraction rules related to to fired redex and the axiom rules of
the variable to be substituted.
\item Moving argument parts of the redex, without adding any rule (contrary to the idempotent intersection frameworks). 
\end{itemize}

\subsection{In rigid derivations}

We assume here that $P \rhd C_\epsi \vdash t:\,\epsi$.\\
The derivation $P'$ defined in \S \ref{subsecSRSE} can obtained
by performing the following transformation at position $a$,
when $a$ ranges over $\Rep_A(b)$.

The subderivation $P\rstr{a}$ must look like:\\
\begin{tikzpicture}
\draw (1.4,3.5) node {$P_0$} ;

\draw (1.7,3.7) node [below right]{
	\begin{prooftree}
	\Infer{0}[ax]{x:\, (S_k)_k \vdash x:\,S_k \posPr{a\cdot 
10\cdot a_k}}
	\Delims{\left( }{ \right)_{k\in K} }
	\end{prooftree}
} ;

\draw [dotted] (1.4,2.2) -- (1.4,3.2) ;
\draw (0,2) -- (6.1,2) ;
	\draw (0,2) node [below right] {$C,\, x:\,(S_k)_{k\in K}
	\vdash r:\,T \posPr{a\cdot 10}$} ;
\draw (0,1.5)--(6.2,1.5) ;
	\draw (0,1.5) node [below right] {$C \vdash \lambda x.r:\, 
	(S_k)_{k\in K}\rew T \posPr{a\cdot 1}$} ;

\draw (9.5,2.5) node [below] { 
	\begin{prooftree}
	\Hypo{P_k }
	\Ellipsis{}{D_k\vdash s:\, S_k \posPr{a\cdot k}}
	\Delims{\left( }{ \right)_{k\in K } }
	\end{prooftree}
} ;

\draw (0,1)--(11.9,1) ;
	\draw (0,1) node [below right] {$C \cup \bigcup\limits_{k\in K} D_k
	\vdash (\lambda x.r)s:\, T \posPr{a}$} ;
\end{tikzpicture}

The subderivation $P\rstr{a}$ must look like:\\

\begin{tikzpicture}
\draw (1.4,3.5) node {$P_0$} ;

\draw (1.7,3.7) node [below right]{
	\begin{prooftree}
	\Infer{0}[\ax]{ x:\, (S_k)_k \vdash x:\,S_k}
	\Delims{\left( }{ \right)_{k\in K} }
	\end{prooftree}
} ;

\draw [dotted] (1.4,2.2) -- (1.4,3.2) ;
\draw (0,2) -- (6.1,2) ;
	\draw (0,2) node [below right] {$C,\, x:\,(S_k)_{k\in K}
	\vdash r:\,T \posPr{a\cdot 10}$} ;
\draw (0,1.5)--(6.2,1.5) ;
	\draw (0,1.5) node [below right] {$C \vdash \lambda x.r:\, 
	(S_k)_{k\in K}\rew T \posPr{a\cdot 1}$} ;

\draw (9.5,2.5) node [below] { 
	\begin{prooftree}
	\Hypo{P_k }
	\Ellipsis{}{D_k\vdash s:\, S_k \posPr{a\cdot k}}
	\Delims{\left( }{ \right)_{k\in K } }
	\end{prooftree}
} ;

\draw (0,1)--(11.9,1) ;
	\draw (0,1) node [below right] {$C \cup \bigcup\limits_{k\in K} D_k
	\vdash (\lambda x.r)s:\, T \posPr{a}$} ;
\end{tikzpicture}

We replace $P\rstr{a}$ by the derivation below:\\

\begin{tikzpicture}
\draw (4.2,3.3) node [below] {
        \begin{prooftree}
        \Hypo{P_k }
        \Ellipsis{}{D_k\vdash s:\, S_k\posPr{a\cdot a_k} }
        \Delims{\left( }{ \right)_{k\in K} }
        \end{prooftree} 
} ;

\draw (1.4,2.5) node {$P_0$} ;
\draw [dotted] (1.4,1.2) -- (1.4,2.2) ;
\draw (0,1)--(5.1,1) ;
        \draw (0,1) node [below right] {$C \cup \bigcup\limits_{k\in K} D_k
        \vdash r[s/x]:\, T \posPr{a}$} ;

\end{tikzpicture}

Notice how this transformation is deterministic: for instance, assume $7 
\in K$. There must be an axiom rule typing $x$ using axiom track 7 : it 
yields $x:\, (S_7)_7\vdash x:\,S_7$ at position $a\cdot 10\cdot a_7$. 
There must also a subderivation at argument track 7: it is $P_7$ at 
position $a\cdot 7$. Then, when we fire the redex at position $b$,
the subderivation $P_7$ \textit{must} replace the axiom rule on track 7,
even if there may be several $i\neq 7$ such that $S_i=S_7$.

\subsection{The quantitative argument}

\label{appQuantArg}

We expose De Carvalho's original argument in system $\mathscr{M}_0$.
We formulate it here w.r.t. rigid derivations.

We work here with \textit{finite} derivations. If $P$ is a finite
derivation, we write $\nr(P)$ ("number of rules") for $|\supp (P)|$.

We want
to show that if the redex at position $b$ is typed -- \textit{i.e.} if
$b\in \ovl{\supp(P)}$--, then the reduced derivation $P'$ verifies
$\nr(P')<\nr(P)$.

With the same notations as previously, we have $\nr(P' \rstr{a})=
\nr(P\rstr{a})-2-|K|$ since the $\symbol{64}$-rule at position $a$
disappear, as well as the $\lambda$-rule at position $a\cdot 1$ and the
$|K|$ axiom rules typing $X$.

Thus, $\nr(P')\leqslant \nr(P)$ as soon as there is a $a\in \supp(P)$
s.t. $\ovl{a}=b$.

\newpage



\section{Equinecessity, Reduction and Approximability}

\label{appEquinecessity}

The rigid construction presented here ensure "trackability", contrary to 
multiset construction of system $\mathscr{M}_0$. We show here a few 
applications useful to prove that approximability is stable under 
reduction or expansion (Lemma \ref{lemApproxRed})
. We consider a \textit{quantitative} derivation 
$P$, with the usual associated notations (including $A=\supp (P)$).

\subsection{Equinecessary bipositions}

\begin{defi} 
Let $\bip_1,\,\bip_2$ two bipositions of $P$.
\begin{itemize}
\item We say $\bip_1$ \textbf{needs} $\bip_2$ if, for all $\supf P\leqslant 
P$, $\bip_1\in \supf P$ implies $\bip_2\in \supf F$.
\item We say $\bip_1$ and $\bip_2$ are
\textbf{equinecessary} (written $\bip_1 \lra \bip_2$)
 if, for all $\supf P\leqslant P$, $\bip_1\in \supf P$ iff
$\bip_2\in \supf P$
\end{itemize}
\end{defi}

There are many elementary equinecessity cases that are easy to observe.
We need only a few one and we define $\up(\bip)$ and $\topb(\bip)$
s.t. $\bip \lra \up(\bip)$ and $\bip \lra \topb(\bip)$ for all $\bip$.
\begin{itemize}
\item $\up(\bip)$ is defined for any $\bip\in \bisupp(P)$ which
is not on an axiom leaf.
\begin{itemize}
	\item $\up(a,\,x,\,k\cdot c)=(a\cdot \ell,\,x,\, k\cdot c)$, where 
$\ell$ is the unique integer s.t. $(a\cdot \ell,\,x,\, k\cdot c)\in
\bisupp(P) $.
	\item If $t(\ovl{a})=\lambda x$, $\up(a,\,\epsi)=(a\cdot 0,\,\epsi),~ 
\up(a,\,1\cdot c)=(a\cdot 0,\, c)$ and $\up(a,\, k \cdot c)=(a\cdot 0,\,
x,\, k\cdot c)$ if $k\geqslant 2$.
	\item If $t(\ovl{a})=\symbol{64}$, $\up(a,\,c)=(a \cdot 1,\, 1\cdot c)$.
\end{itemize}
\item $\topb(\bip)$ is a \textit{right biposition} and is defined by induction for 
any $\bip \in \bisupp (P)$.
\begin{itemize}
	\item If $\bip$ is not on an axiom leaf, then $\topb(\bip)=
\topb(\up (\bip))$.
	\item If $t(\ovl{a})=x$, $\topb(a,\,x,\,k\cdot c)=\top(a,\,c)=(a,\,c)$
\end{itemize}
\end{itemize}
The induction defining is well-founded because of the form of 
the supports of the 001-terms and because $P$ is quantitative.\\

Assume $t\rstr{b}=(\lambda x.r)s$. A very important case of equinecessity
is this one: 
$(a\cdot 1,\, k\cdot c)\lra
(a\cdot 10\cdot a_k,\, x,\, k\cdot c)$ and $(a\cdot 1,\, k\cdot c)
\lra (a\cdot k,\, c)$. Thus, $(a\cdot 10\cdot a_k,\,c) \lra (a\cdot k,\,
c)$.

\subsection{Approximability is stable under (anti)reduction}

We assume here that $P\rew P'$ ($P$ is still assumed to be quantitative).
Let $\supo B\subset \bisupp (P)$ a finite part. We notice 
that $\Res_b$ is defined for any right biposition which is on an axiom leaf 
typing  $y\neq x$.

So, let $\supo \tilde{B}$ be the set obtained from $\topb(B)$ by
replacing any $(a\cdot 10\cdot a_k,\, c)$ by  $(a\cdot k,\, c)$.
Then $|\supo \tilde{B}|\leqslant |\supo B|$ and any $\bip \in 
\supo B$ is equinecessary with a $\tilde{\bip}\in \supo\tilde{B}$.
So the partial proof of Lemma \ref{lemApproxRed} is valid for $\supo
\tilde{B}$. By equinecessity, it entails it is also valid for $\supo B$.

For the converse implication, we just have to replace $\supo B'$ by
$\topb (\supo B')$.

\newpage


\section{Lattices of (finite or not) approximations}

\label{appLat}

\subsection{Types, Forest Types and Contexts}

\label{LatticesTPF}

\begin{defi}
\begin{itemize}
\item Let $U_1$ and $U_2$ two (forest) types. If, as a labelled tree or forest, $U_1$ is a restriction of $U_2$, we write $U_1\leqfty U_2$. When $U_1$ is finite, we write simply $U_1\leqslant U_2$
\item We set $\Approx(U)=\{\supf U ~ |~ \supf U\leqslant U\}$ and 
$\Approx_{\infty} (U)=\{ U_0 ~ |~ U_0\leqfty U\}$
\end{itemize} 
\end{defi}

\begin{lem}
Let $(T_i)_{i\in I}$ be a non-empty family of types, such that $\forall i,\, j\in I,~ \exists T\in \Types,~ T_i,~T_j\leqfty T$ (\textit{i.e.} $T_i,\, T_j$ have an upper bound inside $\Types$).\\
We define the labelled tree $T(I)$ by $\supp (T(I))=\bigcap\limits_{i \in I} \supp (T_i)$ and $T(I)(c)=T_i(c)$ for any $i$.\\
Then, this definition is correct and $T(I)$ is a type (that is finite if one the $T_i$ is). We write $T(I)=\bigwedge\limits_{i\in I} T_i$.
\end{lem}

\begin{proof}
Since $\supp ( T(I))=\bigcap\limits_{i \in I} \supp (T_i)$, $\supp (T(I))$ is a tree without infinite branch ending by $1^\omega$.

Let us assume $c \in \supp(T(I))$. Then, for all $i\in I$, $c \in \supp (T_i)$. For $i,\,j \in I$, there is a $T$ s.t. $T_i,\,T_j\leqfty T$. Thus, $T_i(c)=T(c)=T_j(c)$ and the definition of $T(I)$ is correct.

When $T(I)(c)=\rightarrow$, then $c \in \supp (T_i)$ for all $i\in I$, so $c \cdot 1 \in \supp (T_i)$ for all $i\in I$, so $c \cdot 1 \in \supp (T(I))$, so $T(I) \in \Types$.
\end{proof}

\begin{lem}
Let $(T_i)_{i\in I}$ be a non-empty family of types, such that $\forall i,\, j\in I,~ \exists T\in \Types,~ T_i,~T_j\leqfty T$.\\
We define the labelled tree $T(I)$ by $\supp (T(I))=\bigcup\limits_{i \in I} \supp (T_i)$ and $T(I)(c)=T_i(c)$ for any $i$ s.t. $c \in \supp(T_i)$.\\
Then, this definition is correct and $T(I)$ is a type (that is finite if $I$ is finite and all the $T_i$ are). We write $T(I)=\bigvee\limits_{i\in I} T_i$.
\end{lem}

\begin{proof}
Since $\supp (T(I))=\bigcup\limits_{i \in I} \supp (T_i)$, then $\supp T(I)$ is a
tree.

Let us assume $c \in \supp (T(I))$ and $c \in \supp(T_i)\cap \supp(T_j)$.
Let $T$ be a type s.t. $T_i,\,T_j\leqfty T$. Thus, we have $T_i(c)=T(c)=T_j(c)$ and the definition of $T(I)$ is correct.

Moreover, since $T_i$ is a type, there is a $n\geqslant 0$ s.t. $c \cdot 1^n$ is a leaf of $\supp(T_i)$ and $T_i(c \cdot 1^n)=\alpha$ ($\alpha$ is a type variable). Since $T_i\leqfty T$, $T(c \cdot 1^n)=\alpha$. Since $T$ is a correct type, $T(c\cdot 1^n)=\alpha$ entails that $c\cdot 1^n$ is a leaf of $\supp (T)$ and $T(c\cdot 1^n)=\alpha$. Since $T_j\leqfty T$, $c \cdot 1^n$ is a leaf of $\supp (T_j)$ and $T_j(c \cdot 1^n)=\alpha$. So $c\cdot 1^n$ is a leaf of $T(I)$ and $T(I)(c)=\alpha$ and $\supp (T(I))$ cannot have an infinite branch ending by $1^\omega$.

If moreover $T(I)(c)=\rightarrow$, then $c \in \supp (T_i)$. Since $T_i$ is a correct type, $c \cdot 1 \in \supp (T_i)$ and thus, $c\cdot 1 \in \supp (T(I))$, so $T(I) \in \Types$.
\end{proof}

\begin{prop}
The set $\Types$ endowed with $\leqfty$ is a \textbf{direct complete partial order (d.c.p.o.)}. The join is given by the above operator.\\
Moreover, for any type $T$, $\Approx(T)$ is a distributive lattice and $\Approx_{\infty}(T)$ is a complete distributive lattice, and the meet is given by the above operator.
\end{prop}

\begin{proof}
The distributivity stems from the distributivity of the set-theoretic union and intersection.
\end{proof}

We can likewise construct the joins and the meets of families of forest types (via the set-theoretic operations on the support), provided every pair of elements have an upper bound. The set $\FTypes$ also is a d.c.p.o. and for all f.t. $F$, $\Approx(F)$ is a distributive lattice and $\Approx_{\infty}(F)$ is a complete distributive lattice.

\subsection{A Characterization of Proper Bisupports}

\label{subsecCharBisupp}

Let $B_0\subset \bisupp (P)$. We want to know on what condition $B_0$ is
the support of a derivation $P_0\leqfty P$.

We write $A_0$ for the set of all underlying
outer positions of $\bip$, when $\bip$ spans over $B$.
\begin{itemize}
\item For all $a\in A_0$, we write $T_0(a)$ the labelled tree
induced by $T(a)$ on $\{ c\in \mathbb{N}^*~ |~ (a,~ c) \in B_0\}$.
\item For all $a \in A_0$ and $x\in \mathscr{V}$, we write $C_0(a)(x)$ for the function induced by $C(a)(x)$ on $\{ c\in \mathbb{N}^*~ |~ (a,~ x,~ c)\in B_0\}$.
\end{itemize}

A tedious verification grants that there is a $P_0\leqfty P$ s.t. $\bisupp (P_0)=B_0$ iff the conditions below are satisfied:

\begin{itemize}
\item Support related conditions: $A_0$ is a tree s.t.:
\begin{itemize}
\item For all $a \in A_0$ s.t. $t(\ovl{a})=\symbol{64}$, $a\in
A_0$ implies $(a\cdot 1,~ \epsi)\in B_0$.
\item  For all $a \in A_0$ s.t. $t(\ovl{a})=\lambda x$, $a\in
A_0$ implies $(a,~ \epsi)\in B_0$.\\
\end{itemize}

\item Inner supports related conditions: for all $a \in A,~ T_0(a)$ is a type and for all $x\in \mathscr{V},~ C_0(a)(x)$ is a forest type.\\

\item Axiom rule related conditions: 
for all $x\in \mathscr{V}$,
all $a\in \Ax(x)$ and all $c \in \supp (T(a)),~ 
(a,\, c)\in B_0$ iff $(a,\, x,\, k\cdot c)\in B_0$, where
$k=\tr(a)$.\\

\item Abstraction rule related conditions: for all $x\in \mathscr{V}$,
all $a\in A$ s.t. $t(\ovl{a})=\lambda x$:
\begin{itemize}
\item For all $c \in \mathbb{N}^*,~ (a,\, 1\cdot c)\in B_0$ iff
$(a\cdot 0,\, c)\in B_0$.
\item For all $c \in \mathbb{N}^*$ and all $k\geqslant 2, ~ 
(a,~ k\cdot c)\in B_0$ iff $(a\cdot 0,\, x,\, k\cdot c)\in B_0$.
\item For all $y\in \mathscr{V}-\{x\},~ k\geqslant 2$ and $c \in 
\mathbb{N}^*,~ (a,\, y,\, k\cdot c)\in B_0$ iff 
$(a\cdot 0,\, y,\, k\cdot c) \in B_0$.\\
\end{itemize}

\item Application related conditions:
for all $a\in A$ s.t. $t(\ovl{a})=\symbol{64}$:
\begin{itemize}
\item For all $c \in \mathbb{N}^*,~ (a,~ c)\in B_0$ iff
$(a\cdot 1,~ 1\cdot c)\in B_0$.
\item For all $k\geqslant 2$ and all $c\in \mathbb{N}^*,~ 
(a,~ k\cdot c)\in B_0$ iff $(a\cdot k,~ c)\in B_0$.
\item For all $y\in \mathscr{V},~ k \geqslant 2$ and 
$c \in \mathbb{N}^*$, $(a,~ y,~ k \cdot c)\in B_0$ iff 
$\exists ! \ell \geqslant 1,~ (a\cdot \ell,~ y,~ k\cdot c)\in B_0$.
\end{itemize}
\end{itemize}

\begin{rmk} \label{rmkCharBisupp}
If $P$ is not given, that is, if we have a function 
$P:\, B\rightarrow \mathscr{V}_t \cup \{\rightarrow \}$ where $B$ is 
a set of bipositions (\textit{i.e.} $B\subset \mathbb{N}^*\times \mathbb{N}^*
\cup \mathbb{N}^* \times \mathscr{V}\times \mathbb{N}^*$), on what
condition $P$ is a derivation of $t$ whose support is
$A= \{ a \in \mathbb{N}^* ~ | ~ \text{$a$ is the underlying pos. of a $\bip \in B$} \}$?\\
The above conditions adapts well by replacing $B_0$ by $B$ and $A_0$ by $A$ and adding the following constraints (mostly on labels):
\begin{itemize}
\item $\ovl{A}\subset \supp (t)$.
\item For all $a \in A$ s.t. $t(\ovl{a})=\lambda x$, 
$P(a,\, \epsi)=\rightarrow$ and there is not $k\geqslant 1$ s.t. $a \cdot k\in A$.
\item For all $a \in A$ s.t. $\exists x\in \mathscr{V},~
t(\ovl{a})=x$, 
$C(a)(y)$ is empty for all $y\neq x$ and $\Rt (C(a)(x))$ has
exactly one element.
Thus, for each $a \in A$ s.t. $t(\ovl{a})=x$, we can still 
define $\tr(a)$ as the unique $k$ s.t. $\exists c\in \mathbb{N}^*,~
(a,\, x,\, k\cdot c)\in B$.
\item For all  $a \in A$ s.t. $t(\ovl{a})=\symbol{64}$, $P(a\cdot 1,~ \epsi)=\rew$ and $a\cdot 0 \notin A$. 
\item We must have $P(\bip)=P(\bip')$ for any $\bip$ and $\bip'$ related in
one of the above conditions.
\end{itemize}
\end{rmk}

\subsection{Meets and Joins of Derivations Families}

When $P_0,\, P$ are two derivations typing the same term, we also write $P_0 \leqfty P$ to mean that $P_0$ is the restriction of $P$ on $\bisupp (P_0)$. We set  $\Approx_{\infty} (P)=\{ P_0 \in \Deriv ~ | ~ P_0 \leqfty P\}$.

\begin{lem}
Let $(P_i)_{i\in I}$ be a non-empty family of derivations typing the same term $t$, such that $\forall i,\,j \in I,~ \exists P\in \Deriv,~ P_i,\,P_j\leqfty P$.\\
We define $P(I)$ by $\bisupp (P(I))=\bigcap\limits_{i\in I} \bisupp (P_i)$ and $P(I)(\bip)=P_i(\bip)$ for any $i$.\\
Then, this derivation is correct and the labelled tree $P(I)$ is a derivation (that is finite if one of the $P_i$ is finite). We write $P(I)=\bigwedge\limits_{i\in I} P_i$.
\end{lem}

\begin{proof}
The proof is done by verifying that $P(I)$ satisfies the characterization of
the previous subsection, including Remark \ref{rmkCharBisupp}. It mostly comes to: 
\begin{itemize}
\item The correctness of the definition is granted by the upper bound condition.
\item The definition $P(I)$ grants proper types and contexts, thanks
to subsection \ref{LatticesTPF} .
\item For any $\bip$ and $\bip'$ put at stakes in any of the conditions
of the previous subsection, $\bip \in \bisupp (P(I))$ iff $\forall i\in I,~ \bip \in \bisupp (P_i)$ iff $\forall i \in I,~ \bip'\in  \bisupp (P_i)$ iff $\bip'\in \bisupp (P(I))$.
\item The remaining conditions are proven likewise.
\end{itemize}
\end{proof}

\begin{lem}
Le $(P_i)_{i\in I}$ b a non-empty family of derivations typing the same term, such that $\forall i,\, j\in I,~ \exists P\in \Deriv,~ P_i,~P_j\leqfty P$.\\
We define the labelled tree $P(I)$ by $\bisupp (P(I))=\bigcup\limits_{i \in I} \bisupp (P_i)$ and $P(I)(\bip)=P_i(\bip)$ for any $i$ s.t. $\bip \in \bisupp (P_i)$.\\
Then, this definition is correct and $P(I)$ is a derivation (that is finite if $I$ is finite and all the $P_i$ are). We write $P(I)=\bigvee\limits_{i\in I} P_i$.
\end{lem}

\begin{proof}
The proof is done by verifying that $P(I)$ satisfies the characterization of
the previous subsection, as well as for the previous lemma. But here, for any $\bip$ and $\bip'$ put at stakes in any of the conditions of the previous subsection, $\bip \in \bisupp (P(I))$ iff $\exists i\in I,~ \bip \in \bisupp (P_i)$ iff $\exists i \in I,~ \bip'\in \bisupp (P_i)$ iff $\bip'\in \bisupp (P(I))$.
\end{proof}

The previous lemmas morally define the join and the meet of derivations (under the same derivation) as their set-theoretic union and intersection. More precisely, they entail:


\begin{prop}
The set of derivations typing a same term $t$, endowed with $\leqfty$ is a d.c.p.o. The join of a direct set is given by the above operator.\\
Moreover, for any derivation $P$, $\Approx(P)$ is a distributive lattice (sometimes empty) and $\Approx_{\infty}(P)$ is a complete distributive lattice, and the meet is given by the above operator.
\end{prop}

\subsection{Reach of a derivation}

\label{subsecReach}
\begin{defi}
\begin{itemize}
\item For any derivation $P$, we set $\Reach (P) =\{ \bip \in \bisupp (P)~ |~
\exists  \supf P\leqslant P,~ \bip \in \supf P\}$.
\item If $\bip \in \Reach (P)$, we say $\bip$ is \textbf{reachable}.
\item If $B\subset \bisupp (P)$, we say $B$ is \textbf{reachable} if there is 
$\supf P\leqslant P$ s.t. $B\subset \bisupp  (\supf P)$.
\end{itemize}
\end{defi}

Since $\Approx (P)$ is a complete lattice and the bisupports of its 
elements are finite, we can write $P<\bip>$ (resp. $P<B>$) for the 
smallest $\supf P$ containing $\bip$ (resp. containing $B$), for any 
$\bip \in \Reach (P)$ (resp. for any reachable $B\subset \bisupp 
(P)$).\\

\begin{prop}
Let $B\subset \bisupp (P)$. Then $B$ is reachable iff $B$ is finite and
$B\subset \Reach (P)$.\\
In that case, $P<B>=\bigvee\limits_{\bip \in B} P<\bip>$.
\end{prop}

\begin{defi}
If $\Reach (P)$ is non-empty, we define $P<\Reach>$ as the induced
labelled tree by $P$ on $\Reach (P)$.
\end{defi}

We have actually $P<\Reach> = \bigvee\limits_{\bip \in \Reach (P)} P<\bip>$, so
$P$ is a derivation. By construction, $P$ is approximable.

We can ask ourselves if $P$ is approximable as soon as every biposition
at the root is in its reach. It would lead to a reformulation of the 
approximability condition. We have been unable to answer this question yet.

\subsection{Proof of the subject expansion property}

\label{appExpans}

We reuse the notations and assumptions of \S \ref{secExpans}. We set 
$A'=\supp(P')$. As mentioned in \S \ref{subsecSRSE}, performing an 
expansion of a term inside demands to choose new axiom tracks. We will 
do this \textit{uniformly}, \textit{i.e.} we fix an injection 
$\code{\cdot}$ from $\mathbb{N}^*$ to $\mathbb{N}-\{0,\,1\}$ and any 
axiom rule created at position $a$ will use the axiom track value 
$\code{a}$.

Assume $\supf P'\leqslant P'$. Let $N \in \mathbb{N}$ s.t., for all 
$n\geqslant N,~ b_n \notin \ovl{\supf A'}$ with $\supf A'=\supp (\supf 
P')$. For $n\geqslant N$, we write $\supf P'(n)$ for the derivation 
replacing $t'$ by $t_n$ in $\supf P'$. This derivation is correct 
according to the subject substitution lemma (section \ref{secRed}), 
since $t_n(\ovl{a})=t'(\ovl{a})$ for all $a\in \supf A'$.

Then we write $\supf P'(n,~ k)$ (with $0\leqslant k \leqslant n$) the 
derivation obtained by performing $k$ expansions (w.r.t. our 
reduction sequence and $\code{\cdot}$). Since $b_{n}$ is not in $A$, we 
observe that $\supf P'(n+1,\, 1)= \supf P'(t_n)$. Therefore, for all 
$n\geqslant N,~ \supf P'(n,\, n)=\supf P'(N,\, N)$. Since we could 
replace $N$ by any $n\geqslant N$, $\supf P$ is morally $\supf 
P'(\infty,\, \infty)$. We write $P=P'(\init)$ to refer to this 
deterministic construction.\\


We set $\mathscr{D}=\{\supf P'(\init)~ |~ \supf P'\leqslant P'\}$. Let us 
show that $\mathscr{D}$ is a directed set.

Let $\supf P'_1,\, \supf P'_2\leqslant P'$. We set 
$\supf P'=\supf P'_1 \vee \supf P'_2$.             
Let $N$ be great enough so that $\forall n\geqslant N,~ b_n \notin 
\ovl{\supf A'}$ with $\supf A'=\supp (\supf P')$.

We have $\supf P'_i\leqslant \supf P'$, so $\supf P_i'(N) \leqslant 
\supf P(N)$, so, the by monotonicity of \textit{uniform} expansion, 
$\supf P'_i(N,\,N)\leqslant \supf P'(N,\,N)$, i.e. 
$\supf P_i(\init)\leqslant \supf P(\init)$.

Since $\mathscr{D}$ is directed, we can set 
$P=\bigvee\limits_{\supf P'\leqslant P'} \supf P'(\init)$.
Since for any $\supf P'\leqslant P$ and the associated usual notation, 
$\supf C(\epsi)=\supf C'(\epsi),~ \supf T(\epsi)=\supf T'(\epsi)$ and 
$C(\epsi),~ C'(\epsi),~ T(\epsi),~ T'(\epsi)$ are the respective
infinite join of $\supf C(\epsi),~ \supf C'(\epsi),~ 
\supf T(\epsi),~ \supf T'(\epsi)$ when $\supf P'$ ranges over $\Approx(P')$,
we conclude that $C(\epsi)=C'(\epsi)$ and $T(\epsi)=T'(\epsi)$.

We can also prove that different choices of coding functions $\code{\cdot}$ 
yield \textit{isomorphic} derivation typing $t$. We start by proving it when 
$P'$ is finite.

\newpage
\section{Approximability of the quantitative NF-derivations}

\label{appNF}

We show in this appendix that every quantitative derivation typing a 
normal form $t$ is approximable. We use the same notations as in Section 
\ref{secNF}: we consider a derivation $P$ built as in Subsection 
\ref{subsecSpCons}, from a normal form $t$, a d-support $A$ of $t$ and a 
type $T(a)$ given for each full position of $A$. It yields a family of 
contexts $(C(a))_{a\in A}$ and of types $(T(a))_{a\in A}$ such that 
$P(a)$ is $C(a)\vdash t|_{\overline{a}}: \,T(a)$ for all $a\in A$.

 

\subsection{Degree of a position inside a type in a derivation}

\textbullet~ For each $a$ in $A$ and each position $c$ in $S(a)$ such that
$S(a)(c)\neq X_i$, we define the number $\depth_s(c)$ by:
\begin{itemize}
\item When $a$ is a full node, $\depth_s(c)$ is the applicative depth of $a$.
\item When $a$ is an abstraction position: the value of $\depth_s(c)$ for
the positions colored in red is the applicative depth of $a$
(and of $\mathring{a}$).
$$E(a)(x_1)\red{\rightarrow} E(a)(x_2)\red{\rightarrow} \ldots \red{\rightarrow}
E(a)(x_n)\red{\rightarrow T(\mathring{a})}$$
\item When $a$ is partial: the value of $\depth_s(c)$ for the positions 
colored in red is the applicative depth of $a$.
$$F_1(a)\red{\rightarrow} \ldots
\red{\rightarrow} F_k(a)\red{\rightarrow T(\mathring{a})}$$
\end{itemize}

For each $a\in A$ and each position $c$ in $T(a)$, we define the number
$\depth_s(c)$ (that is the applicative depth of the position $a'$ on which
$c$ depends) by extending $\depth_s$ via substitution.

There again, for each $a\in A$, each variable $x$ and each position
$c$ in $C(a)(x)(c)$ we define $\depth_s(c)$ by extending $\depth_s$ via 
substitution.

The definition of $\depth_s(c)$ in $T(a)$ and $C(a)(x)$ is sound, because in
$E(a)(x)$ and $F_k(x)$, there a no symbol other than the $X_i$. But
the $X_i$ disappear thanks to the coinductive definition of $T(a)$: every
position $c$ of $T(a)$ will receive a value for $\depth_s(c)$.

\textbullet~ For $c\in \mathbb{N}^*$, we set $s(c)=\max 
(\ell,~ c_0,~ c_1,\ldots,~ c_{n-1})$ 
where $n=|a|$ and $\ell=|\{ 0 \leqslant i \leqslant n~ |~ c_i\geqslant 2\}|$.\\
For each $a$ in $A$ and each position $c$ in $S(a)$ such that
$S(a)(c)\neq X_i$, we define the number $\depth_i(c)$ by:
\begin{itemize}
\item When $a$ is a full node: $\depth_i(c)=s(c)$ ($c$ is a position of 
$T(a)$).
\item When $a$ is an abstraction node:
if $c$ is a position colored in blue, $\depth_i(c)=0$ and if $c$ is
colored in red, $\depth_i(c)=s(c')$ (when $c'$ is the position corresponding to
$c$ in $T(\mathring{a})$, \textit{i.e.} $c=0^k\cdot c')$
$$E(a)(x_1)\blue{\rightarrow} E(a)(x_2)
\blue{\rightarrow} \ldots \blue{\rightarrow}
E(a)(x_k)\blue{\rightarrow} \red{T(\mathring{a})}$$
\item  When $a$ is partial:
if $c$ is a position colored in blue, $\depth_i(c)=0$ and if $c$ is
colored in red, $\depth_i(c)=s(c')$ (when $c'$ is the position corresponding 
to $c$ in $T(\mathring{a}$), \textit{i.e.} $c=0^k\cdot c')$
$$F_1(a)\blue{\rightarrow} ...
\blue{\rightarrow} F_k(x)\blue{\rightarrow}\red{T(\mathring{a})}$$
\end{itemize}

We extend likewise $n(c)$ for inner positions of $T(a)$ or in $C(a)(x)$
via substitution.

\begin{defi}
If $c$ is a position in $T(a)$ or in $C(a)(x)$, the \textbf{degree} of $c$ is
defined by $\deg c=\max(\depth_i(c),~ \depth_i(c))$.
\end{defi}

\subsection{More formally...}

For $c\in \mathbb{N}^*$, we set 
$s(c)=\max 
(\ell,~ c_0,~ c_1,\ldots,~ c_{n-1})$
where $n=|a|$ and $\ell=|\{ 0 \leqslant i \leqslant n~ |~ c_i\geqslant 2\}|$.

For all $a\in A$ and $k\in \mathbb{N}$, we set $S\supo (a)=X_a$ and 
$S^{k+1}(a)= S^k(a)[S(a')/X_{a'}]_{a'\in \mathbb{N}^*}$.

For all $k \in \mathbb{N}$, we set $\supp^* (S^k(a))=\{c\in \supp (S^k(a))~ |~ 
S^k(a)(c)\neq X_{a'}\}$.\\

If $c\in \supp (T(a))$, there is a minimal $n$ s.t. $c\in \supp^* (S^k(a))$.
We denote it $\cd(a)(c)$ (call-depth of $c$ at pos. $a$).

In that case, there are unique $c'\in \supp (T(a))$ and $a'\in A$ s.t.
$c'\leqslant c,~ S^{k-1}(a)(c')=X_{a'}$ (we have necessarily $a\leqslant a'$).
We write $a'=\spo(a)(c)$ (source position of $c$ at pos. $a$)
and $c'=\lcp(a)(c)$ (last calling position of $c$ at pos. $a$).\\
Then, we set $\depth_s(a)(c):= \ad{a'}$.\\

With the same notations, $T(a')$ is of the shape $C(a_1)(x_1)\rightarrow 
\ldots C(a_k)(x_k) \rew C(\mathring{a}')$ or $\Rft_1(a')\rightarrow 
\Rft_2(a')\rightarrow \ldots \rew \Rft_k(a') \rightarrow T(\mathring{a'})$, where the forest type $\Rft_k(a')$ is $(\code{a_0}\cdot T(a_0))_{a_0 \in AP_k(a')}$.
There are two cases:
\begin{itemize}
\item $c=c'\cdot 0^j$ with $j< k$: we set $n(a)(c)=0$.
\item $c=c'\cdot 0^k:c"$ (with $c"\in \supp (T(\mathring{a'}))$): we write 
$\sip (c)= c"$ (source inner position of $c$ at pos. $a$) and set 
$\depth_i(a)(c)= s(c")$.
\end{itemize}


\begin{defi}
If $c$ is a position in $T(a)$ or in $C(a)(x)$, the \textbf{degree} of $c$ at 
pos. $a$ is defined by $\deg (a,~ c)=\max(\depth_s(c),\, \depth_i(c))$.
\end{defi}

\begin{lem} \label{CuttingLemma}
For all $k\in \mathbb{N}$ and all $a,~ c\in \mathbb{N}^*$, we have $a\in A_n$ and 
$c\in \supp (S_n^k(a))$ iff $a\in A,~ c\in \supp (S^k(a))$ and
$\deg (a,~ c) \leqslant n$. \\
In that case, $S^k(a)(c)=S_n^k(a)(c)$.
\end{lem}

\begin{proof}
By a simple but tedious induction on $k$.\\

\begin{itemize}
\item Case $k=0$: \\
If $a\in A_n$ and $c\in \supp (S_n^0(a))$, then $\ad{a}\leqslant n$ (by def. of
$A_n$) and $c=\epsi$. 
By definition, $\depth_s(a)(\epsi)=\ad{a}\leqslant n$ and 
$\depth_i(a)(\epsi)=0$. Thus, $\deg (a,~ c)\leqslant n$.\\
Conversely, if $c\in \supp (S^0(a))$ and $\deg (a,~ \epsi)\leqslant n$, 
we have likewise $c=\epsi$ and $\spo(a)(\epsi)=a$, so $\ad{a}\leqslant 
n$. Thus, $a\in A_n$ and then $c\in S_n^0(a)$.\\
\item Case $k+1$:\\
If $a\in A_n$ and $c\in \supp (S_n^{k+1}(a))$, we assume that $a\notin 
\supp (S_n^k(a))$ (case already handled by IH). We set $a'=\spo_n
(a)(c)$ and $c'=\lcp_n(a)(c)$ (thus, $S_n^k(a)(c')=X_{a'}$). 
By IH, we have also $a'\in A,~ a'=\spo(a)(c)$ and $c'=\lcp(a)(c)$. 
We have two subcases, depending if $S_n^{k+1}(a)(c)=X_{a"}$ holds or not.\\
\begin{itemize}
\item If $S_n^{k+1}(a)(c)=X_{a"}$ (with necessarily $a"\in A_n$), 
then $c=c'\cdot 0^j:\ell$ with $j< \rdeg(a')$ and $\ell$ integer
and, by IH, $c'\in \supp (S^k(a))$ and $S^k(a)(c')=X_{a'}$.\\
Then $c=c'\cdot 0^j:\ell \in \supp (S^{k+1}(a)),~ \depth_s(a)(c)=
\ad{a"}\leqslant n$ (since $a"\in A_n)$ 
and $\depth_i(a)(c)=0$. So we have $\deg (a, c)\leqslant n$. \\
\item  If $S_n^{k+1}(a)(c)\neq X_{a"}$ for all $a"$, then $c=c'\cdot 0^j$ with 
$j< \rdeg(a')$ or $c=c'\cdot 0^{\rdeg(a')}\cdot c"$ with 
$c"\in \supp (T_n(\mathring{a'}))\subset \supp (T_n(\mathring{a'}))$.\\
In both cases, $c\in \supp (S^{k+1}(a))$. In the former one, $\depth_s(a)(c)=
0$ and in the latter one, $\depth_s(c)=s(c")\leqslant n$ 
(because $c"\in \supp (T_n(\mathring{a'}))$). Therefore, $\deg (a,~ c) 
\leqslant n$.\\
\end{itemize}
Conversely, if $a\in A,~ c\in \supp (S^{k+1}(a))$ and $\deg(a,~ c)\leqslant n$, we assume that $a\notin \supp (S^k(a))$ (case already handled by IH). We set $a'=\spo (a)(c)$ and $c'=\lcp(a)(c)$ (thus, $S^k(a)(c')=X_{a'}$). By IH, we have also $a,~ a'\in A_n,~ a'=\spo_n(a)(c)$ and $c'=\lcp_n(a)(c)$. Likewise, we have two subcases, according to whether $S^{k+1}(a)(c)=X_{a"}$ or not.
\begin{itemize}
\item If $S^{k+1}(a)(c)=X_{a"}$, then, by def. of $\depth_i$, we have
$\depth_i(a)(c)=\ad{a"}$, so $\ad{a"}\leqslant \deg (a,~ c) \leqslant  n$, so
$a"\in A_n$. Since $a\leqslant a'$, $a\in A_n$.\\
Moreover, $c=c'\cdot 0^j\cdot \ell$ with $j<\rdeg(a')$ and $\ell$ integer.
Since $a"\in A_n$, we have also $c\in S_n^{k+1}(a)$ and $S_n^{k+1}(a)(c)=X_{a"}$.\\

\item If $S_n^{k+1}(a)(c)\neq X_{a"}$ for all $a"$, then $c=c'\cdot 0^j$ with 
$j< \rdeg(a')$ or $c=c'\cdot 0^{\rdeg(a')}\cdot c"$ with $\supp 
(T(\mathring{a'}))$.\\
Since $\depth_s(c)=s(c")$, $s(c")\leqslant n$, so $c"\in \supp 
(T_n(\mathring{a'}))$. Thus, $c\in \supp (S_n^{k+1}(a))$ and $S_n^{k+1}(a)(c)
=S^{k+1}(a)(c)$.
\end{itemize}
\end{itemize}
\end{proof}

\subsection{A complete sequence of derivation approximations}

Let $n$ be an integer.
\begin{itemize}
\item We set $A_n=\{a\in A~ |~ \text{the applicative depth of
$a$ is $\leqslant n$}\}$.
\item We define $T_n(a)$ and $C_n(a)(x)$ by
removing all positions $c$ such that $\deg c > n+1$.
\item We define the finite labelled tree $P_n$ by $\supp (P_n)=A_n$ and,
for each $a\in A_n,~ P(a)= C_n(a) \vdash t|_a:\,T_n(a)$.
\end{itemize}

\begin{prop}
The labelled tree $P_n$ is a finite derivation and $P_n\leqslant P$. It is
actually the derivation obtained by the trivial construction w.r.t.
$T_n$ and $A_n$.
\end{prop}

\begin{proof}
It is a straightaway consequence of lemma \ref{CuttingLemma}.\\
We use the notations of the previous subsection and write $\tilde{T} _n ,~ \tilde{C} _n,~ \tilde{P}_n $ for the type, context and derivation obtained by the trivial construction based w.r.t. $(A_n,~ T_n)$.

Since $\tilde{T}_n (a)=S_n(a)[\tilde{T}_n(a')/X_{a'}]_{a'\in \mathbb{N}^*}$, 
if $c\in \supp (\tilde{T}_n(a))$, there is $k$ s.t. $c\in \supp (S_n^k(a))$ and
$S_n^k(a)(c)=\tilde{T}_n(a)$. By lemma \ref{CuttingLemma}, we have
also $c\in \supp (S^k(a)),~ \deg (a,~ c)\leqslant n$ and $S^k(a)(c)=S_n^k(a)(c)$.
Thus, $T_n(a)(c)=\tilde{T}_n(a)(c)$. Conversely, we show likewise that
if $c\in \supp (T_n(a))$, then $c\in \supp (\tilde{T}_n(a))$.

Thus, for all $a\in A_n,~ \tilde{T}_n(a)=T_n(a)$. 
It also entails that $C_n(a)(x)=\tilde{C}_n(a)(x)$ for all $a\in A_n$ and 
$x\in \mathscr{V}$.
\end{proof}

\begin{corol}
The derivation $P$ is approximable.
\end{corol}

\begin{proof}
Let $\supo B\subset \bisupp (P)$ be a finite set. Let $n$ be the
maximal degree of a biposition of $B$. Then, $\supo B\subset
\bisupp (P_n)$ and $P_n \leqslant P$ is finite.
\end{proof}

\begin{corol}  \label{NFNDTypable}
\begin{itemize}
\item Each normal form $t\in \Lambda^{001}$ admits a approximable and 
unforgetful derivation.
\item Every quantitative derivation typing a normal form is approximable.
\end{itemize}
\end{corol}

\begin{proof}
\begin{itemize}
\item Comes from the previous corollary and Proposition \ref{propNFhasNDD}.
\item Comes from the previous corollary and Proposition \ref{propCharQTofNF}.
\end{itemize} 
\end{proof}

\newpage
\section{Isomorphisms between rigid derivations}

\label{appIso}

Let $P_1$ and $P_2$ be two rigid derivations typing the same term
$t$. We write $A_i,~  C_i,~ T_i$ for their respective supports,
contexts and types.

A \textbf{derivation isomorphism} $\phi$ from $P_1$ to $P_2$ is given by:
\begin{itemize}
\item $\phi_{\supp}$, 01-tree isomorphism from $A_1$ to $A_2$. 
\item For each $a_1 \in A_1$:
\begin{itemize}
\item A type isomorphism $\phi_{a_1}:~ T_1(a_1)\rew T_2(\phi_{\supp}(a_1))$
\item For each $x\in \mathscr{V}$, a forest type isomorphism 
$\phi_{a_1|x}:\, C_1(a_1)(x)\rightarrow C_2(\phi_{\supp}(a_1))(x)$
\end{itemize} 
\end{itemize} 
such that the following "rules compatibility" conditions hold:
\begin{itemize}
\item If $t(\ovl{a_1})=\lambda x$, then:
\begin{itemize}
\item $\phi_{a_1}(1\cdot c)=1\cdot \phi_{a_1\cdot 0}(c)$ and 
$\phi_{a_1}(k\cdot c)=\phi_{a_1\cdot 1|x}(k\cdot c)$ for any
$k\geqslant 2$ and $c\in \mathbb{N}^*$
\item $\phi_{a_1|y}=\phi_{a_1\cdot 0|y}$ for any $y\in \mathscr{V},~ 
y\neq x$.
\end{itemize}
\item If $t(\ovl{a_1})=\symbol{64}$:
\begin{itemize}
\item $\phi_{a_1}(c)=\Tl(\phi_{a_1\cdot 1}(1\cdot c))$, for any $c\in
\mathbb{N}^*$, where $\Tl(k\cdot c)=c$ (removal of the first integer
in a finite sequence).
\item $\phi_{a_1|x}=\bigcup\limits_{\ell \geqslant 1} \phi_{a_1\cdot \ell}$
(the functional join must be defined because of the app-rule).
\end{itemize}
\end{itemize}

The above rules means that $\phi$ must respect different occurrences of 
the "same" (from a moral point of view) biposition. For instance: 
\begin{itemize}
\item  Assume $t(\ovl{a_1})=\lambda x$, then $T_1(a_1)=C_1(a_1\cdot 0)(x)\rew 
T_1(a_1\cdot 0)$. So, any inner position $c_1$ inside $T(a_1\cdot 0)$ 
can be "identified" to the inner position $1\cdot c_1$ inside $T_1(a_1)$. 
Thus (forgetting about the indexes), if $\phi$ maps $c_1$ on $c_2$ (inside
$T_2(a_2)$, then $\phi$ should map $1\cdot c_1$ on $1\cdot c_2$.
\item Assume $t(\ovl{a_1})=\symbol{64}$. Then, the forest type 
$C(a_1)(x)$ if the union of the $C(a_1\cdot \ell)(x)$ (for $\ell$ spanning
over $\mathbb{N}-\{0\}$). Then $\phi$ should map every inner position
$k\cdot c$ inside $C(a_1)(x)$ according to the unique $C(a_1\cdot \ell)$ 
to which it belong.
\end{itemize}

\begin{lemma}
If $P_1\rewb{b} P_1',~ P_2\rewb{b} P_2'$, then $P_1\equiv P_2$
iff $P_1 '\equiv P_2$.
\end{lemma}

\begin{proof} 
Let $\alpha'_1\in A'_1$. We set $\alpha_1 = \Res_b^{-1}(\alpha'_1)$,
$\alpha_2=\phi_{\supp (\alpha_1)},~ \alpha_2'=\Res_b(\alpha_2)$
($\Res_b$ is meant w.r.t. $P_1$ or $P_2$ according to the cases).
Then we set $\phi'_{\supp}=\Res_b\circ \phi_{\supp} \circ \Res_b^{-1}$.
Thus, $\alpha_2'= \phi '_{\supp}(\alpha'_1)$. 

We set $\phi'_{\alpha' _1}=\phi_{\alpha_1}$.
Observing the form of $C_1(\alpha_1)(y)$ (for $y \neq x$) given in 
Subsection \ref{subsecResBip}, we set $\phi'_{\alpha'_1|y}=\phi_{\alpha_1|y}\cup 
\bigcup\limits_{k \in K} \phi_{a(k)|x}$ with $K=\AxTr(\alpha_1,\, 
x,\,k)$ and $a(k)=\pos(a_1,\,x,\, k)$.
\end{proof}

Notice that $\phi'$ is defined deterministically from $\phi$.

\begin{prop} 
If $P_1$ and $P_2$ are isomorphic and type the term $t$ (we do not assume 
them to be approximable), $t \rew^\infty t'$, yielding two derivation
$P_1 ',~ P_2$ according to section \ref{secRed}, then $P_1'$ and $P_2 '$
are also isomorphism.
\end{prop}

\begin{proof}
We reuse all the hypotheses and notations of section \ref{secRed} and
we consider an isomorphism $\phi:\, P_1\rew P_2$.

For all $n\in \mathbb{N}$, let $P_1^n,~ P_2^n$ and $\phi^n$ be the 
derivations and derivation isomorphisms obtained after $n$ steps of 
reduction from $P_1,~ P_2$ and $\phi$. Let $\alpha'_1\in A'_1$ and 
$N\in \mathbb{N}$ such that, for 
all $n\geqslant N,~ |b_n|>|\alpha'_1|$. But then, for any $n\geqslant 
N$, $C_i^n(\alpha')(x)=C'_i\alpha')(x),\, T_i'(\alpha')=T_i^n(\alpha')$. 
So we can set $\phi'_{\supp}(\alpha_1)= \phi_{\supp}^n(\alpha_1)$, 
$\phi'_{\alpha'} = \phi^N_{\alpha'}$.
\end{proof}

\section{An Infinitary Type System with Multiset Constructions}

\label{appM} 

\subsection{Rules}

We present here a definition of type assignment system $\mathscr{M}$,
which is an infinitary version of De Carvalho's system $\mathscr{M}_0$.

If two (forest) types $U_1$ and $U_2$ are isomorphic, we write $U_1 \equiv U_2$.
The set $\Types_{\mathscr{M}}$ is the set $\Types/\equiv$ and the set 
$\mathscr{M}(\Types)$ of  multiset types is defined as $\FTypes/\equiv $.

If $U$ is a forest or a rigid type, its equivalence class is written 
$\overline{U}$. The equivalent class of a forest type $F$ is the 
multiset type written $[\overline{F_{|k}}]_{k\in \Rt (F)}$ and the one 
of the rigid type $F\rightarrow T$ is the type $\overline{F}\rew 
\overline{T}$. If $\alpha$ is a type variable, $\overline{\alpha}$ is 
written simply $\alpha$ (instead of $\{\alpha \}$). It defines 
coinductively the multiset style writing of $\overline{U}$.

Countable sum $\sum\limits_{i\in I} \overline{F^{i}}$ is defined on 
$\mathscr{M}(\Types)$ by using a bijection $j$ from the pairwise 
disjoint countable sum $\coprod\limits_{\mathbb{N}} 
\mathbb{N}-\{0,\,1\}$ to $\mathbb{N}-\{0,\,1\}$, replacing $I$ by a part 
of $\mathbb{N}$ and each root $k$ of $F^{i}$ by the integer given by 
$j(i,\,k)$ (so that the equivalence classes are preserved).

A $\mathscr{M}$-context is a function from the set of term variables 
$\mathscr{V}$ to the set $\Types_\mathscr{M}$. The set of $\ast$-derivations, 
written $\Deriv_*$ is defined coinductively by the following rules:

\begin{center}
\begin{prooftree}
\Infer{0}[ax]{x:\,[\tau] \vdash x:\,\tau \posPr{\epsi} }
\end{prooftree}
\hspace{3cm}
\begin{prooftree}
\Hypo{P'}
\Infer{1}{\Gamma \vdash t:\, \tau \posPr{0}}
\Infer{1}[abs]{\Gamma-x \vdash \lambda x.t:~
\Gamma(x)\rightarrow \tau \posPr{\epsi}}
\end{prooftree}\\[1cm]

\begin{prooftree}
\Hypo{P'}
\Infer{1}{\Gamma \vdash t:\, [\sigma_i]_{i \in I}\rightarrow \tau
\posPr{1} }
\Hypo{P_k '}
\Infer{1}{\Delta_i \vdash u:\, \sigma_i ~  \posPr{k_i}}
\Delims{ \left( }{ \right)_{i\in I} }
\Infer{2}[app]{\Gamma + \sum\limits_{i\in I} \Delta_i \vdash t(u):\, \tau
\posPr{\epsi}}
\end{prooftree}
\end{center}

In the app-rule, the $k_i$ must be pairwise distinct integers $\geqslant 2$.

Let $P_1$ and $P_2$ be two $\ast$-derivations.  A $\ast$-isomorphism
from $P_1$ to $P_2$ is a 01-labelled isomorphism from $P_1$ to $P_2$
and the set $\Deriv_{\mathscr{M}}$ is defined by $\Deriv_*/\equiv$.

From now on, we write $\Types$ and $\Deriv$ instead of
$\Types_{\mathscr{M}}$ and $\Deriv_{\mathscr{M}}$.  An element of
$\Deriv$ is usually written $\Pi$, whereas an element of $\Deriv_*$ is
written $P$. Notice the derivation $\Pi$ and $\Pi'$ of Subsection
\ref{subsecExamples} are objects of $\Deriv$.

\subsection{Quantitativity and Coinduction}

\label{subsecUnquant}

Let $\Gamma$ be any context. Using the infinite branch of $\fom$,
we notice we can give the following variant of derivation $\Pi'$ 
(subsection \ref{subsecExamples}), which still respects the rules
of system $\mathscr{M}$:

\begin{center}
\begin{prooftree}
\Infer{0}[\text{ax}]{f:\,[[\alpha]\rightarrow \alpha] \vdash
\fom:\,[\alpha]\rightarrow \alpha}
\Hypo{\Pi'_{\Gamma}}
\Infer{1}{f:[[\alpha]\rightarrow \alpha]_{n\in \omega} + \Gamma \vdash
\fom:\,\alpha}
\Infer{2}[app]{f:[[\alpha]\rightarrow \alpha]_{n\in
\omega} +\Gamma \vdash \fom :\,\alpha}
\end{prooftree}
\end{center}

If, for instance, we choose the context $\Gamma$ to be $x:\,\tau$, from 
a quantitative point of view, the variable $x$ (that is not in the typed 
term $\fom$) should not morally be present in the context. We have been 
able to "call" the type $\tau$ by the mean of an infinite branch. Thus, 
we can enrich the type of any variable in any part of a derivation, as 
long it is below an infinite branch (neglecting the bound variables).
It motivates the following definition:

\begin{defi}
\begin{itemize}
\item A $\ast$-derivation $P$ is \textbf{quantitative} if, for all $a\in \supp(P),~  \Gamma(a)(x)= [\tau(a')]_{a'\in \Ax(a)(x)}$.
\item A derivation $\Pi$ is quantitative if any of its $\ast$-representatives is (in that case, all of them are quantitative).
\end{itemize}
\end{defi}

In the next subsection, we show that a derivation $\Pi$ from system
$\mathscr{M}$ can have both quantitative and unquantitative
representatives in the rigid framework. It once again shows that rigid
constructions allow a more fine-grained control than system
$\mathscr{M}$ does on derivations.

\subsection{Representatives and Dynamics}

A rigid derivation $P$ (with the usual notations $C,~ t,~ T$) 
\textbf{represents} a derivation $\Pi$ if the $\ast$-derivation $P_*$
defined by $\supp (P_*)=\supp (P)$ and $P_*(a)=\overline{C(a)}\vdash
t|_{\overline{a}}:\,\overline{T(a)}$, is a representative of $\Pi$.
We write $P_1 \eqm P_2$ when $P_1$ and $P_2$ both represent the same
derivation $\Pi$. 

\begin{prop}
If a rigid derivation $P$ is quantitative, then the derivation 
$\ovl{P}$ (in system $\mathscr{M}$) is quantitative.
\end{prop}

Proposition \ref{propCharQTofNF} makes easy to prove that:

\begin{prop}
If $\Pi$ is a quantitative derivation typing a normal form, then, there is 
a quantitative rigid derivation $P$ s.t. $\overline{P}=\Pi$.
\end{prop}

\begin{proof}
Let $P(\ast)$ be a $\ast$-derivation representing $\Pi$. We set $A=\supp(P(\ast))$ and for all full position $a\in A$, we choose a representative $T(\mathring{a})$ of $\tau(a)$. We apply then the special construction, which yields a rigid derivation $P$ such that $P_*=P(\ast)$ (we show that, for all $a\in A$, $T(a)$ represents $\tau(a)$).
  \end{proof}

We can actually prove that every quantitative derivation can be 
represented with a quantitative rigid derivation and that we can endow 
it with every possible infinitary reduction choice (\cite{vialXX}). 
However, a quantitative derivation can also have an unquantitative rigid 
representative (see below $\Pi'$ and $\tilde{P}'$).

Actually, whereas $P_1\equiv P_2$ (Subsection \ref{subsecDefRD}) means
that $P_1$ and $P_2$ are isomorphic in every possible way, $P_1 \eqm P_2$ is far weaker: we explicit in this subsection 
big differences in the dynamical behaviour
between two rigid representatives of the derivations $\Pi$ and $\Pi'$
of Subsection \ref{subsecExamples}.

We omit the right side of axiom rules, \textit{e.g.}
$f:\,((2\cdot \alpha)\rightarrow \alpha)_2 $ stands for
$f:\,((2\cdot \alpha)\rightarrow \alpha)_2 \vdash
f:\,(2\cdot \alpha)\rightarrow \alpha$.\\

\noindent \textbullet~ Let $P_k$ ($k\geqslant 2$) and $P$ be the following 
rigid derivations:
{\footnotesize
\begin{center}
$P_k=\hspace{0.7cm}$\begin{prooftree}
\Infer{0}{f:\,((2\cdot \alpha)\rightarrow \alpha)_k\vdash 
\trck{1} }
\Infer{0}{x:\,(\gamma)_2 \vdash x:\,\gamma \trck{1}}
\Infer{0}{x:(\gamma)_i\vdash x:\,\gamma \trck{i-1}}
\Delims{ \left( }{ \right)_{ i \geqslant 4 } }
\Infer{2}{x:(\gamma)_{i\geqslant 2} \vdash xx:\,\alpha \trck{2}}
\Infer{2}{f:((2\cdot\alpha)\rightarrow \alpha)_k \vdash f(xx):\,\alpha\trck{0}}
\Infer{1}{f:((2:\alpha)\rightarrow \alpha)_k \vdash \Delta_f:\,\gamma}
\end{prooftree}\\
\end{center}

\begin{center}
$P=\hspace{0.7cm}$\begin{prooftree}
\Hypo{P_2 \trck{1}}
\Hypo{P_k \trck{k-1} }
\Delims{ \left( }{ \right)_{ k \geqslant 3 } }
\Infer{2}{f:((2:\alpha)\rightarrow \alpha)_{k\geqslant 2} \vdash
\Delta_f \Delta_f}
\end{prooftree}\\
\end{center}
} 

\noindent \textbullet ~ Let $\tilde{P}_k$ ($k\geqslant 2$) and $\tilde{P}$ be the following rigid derivations:

{\footnotesize
\begin{center}
$\tilde{P}_k=\hspace{0.7cm}$\begin{prooftree}
\Infer{0}{f:((2\cdot \alpha)\rightarrow \alpha)_k  \trck{1}}
\Infer{0}{x:(\gamma)_3 \vdash x:\,\gamma \trck{1}}
\Infer{0}{x:(\gamma)_2 \vdash x:\,\gamma \trck{2}}
\Infer{0}{x:(\gamma)_i \vdash x:\,\gamma \trck{i-1}}
\Delims{ \left( }{ \right)_{ i \geqslant 4 } }
\Infer{3}{x:(\gamma)_{i\geqslant 2} \vdash xx:\,\alpha \trck{2}}
\Infer{2}{f:((2\cdot\alpha)\rightarrow \alpha)_k \vdash f(xx):\,\alpha \trck{0}}
\Infer{1}{f:((2\cdot \alpha)\rightarrow \alpha)_k \vdash \Delta_f:\,\gamma}
\end{prooftree}\\
\end{center}

\begin{center}
$\tilde{P}=\hspace{0.7cm}$\begin{prooftree}
\Hypo{\tilde{P}_2 \trck{1}}
\Hypo{\tilde{P}_k \trck{k-1} }
\Delims{ \left( }{ \right)_{ k \geqslant 3 } }
\Infer{2}{f:((2\cdot \alpha)\rightarrow \alpha)_{k\geqslant 2} \vdash
\Delta_f \Delta_f}
\end{prooftree}\\
\end{center}}

\noindent \textbullet ~ The rigid derivations $P$ and $\tilde{P}$ both represent
$\Pi$. Morally, subject reduction in $P$ will consist in taking the
first argument $P_3$, placing it on the first occurrence of $x$ in
$f(xx)$ (in $P_2$) and putting the other $P_k$ ($k\geqslant 4$) in the
different axiom rules typing the second occurrence of $x$ in the same
order. There is a simple decrease on the track number and we can
go this way towards $f^\omega$.

The rigid derivation $\tilde{P}$ process the same way, except it will always skip $\tilde{P}_3$ ($\tilde{P}_3$) will stay on track 2). Morally, we 
perform subject reduction "by-hand" while avoiding to ever place $P_3$ in head position.

The definitions of section \ref{secRed} show that infinitary reductions performed in $P$ and $\tilde{P}$ yield respectively to $P'$ and $\tilde{P}'$ below.

\begin{center}
$P'=$
{\footnotesize
\begin{prooftree} 
\Infer{0}[ax]{
f:\,((2\cdot \alpha)\rightarrow \alpha)_2
 \trck{1}}

\Infer{0}[ax]{
f:\,((2\cdot \alpha)\rightarrow \alpha)_3  \trck{1}}

\Hypo{P'}
\Ellipsis{}{f:\,((2\cdot \alpha)\rightarrow \alpha)_{k\geqslant 4}\vdash
f^\omega:\, \alpha \trck{2}}

\Infer{2}[app]{ f:\,((2\cdot \alpha)\rightarrow \alpha)_{k\geqslant 3}\vdash
f^\omega:\, \alpha \trck{2} }

\Infer{2}[app]{ f:\,((2\cdot \alpha)\rightarrow \alpha)_{k\geqslant 2}\vdash
f^\omega:\, \alpha}
\end{prooftree}
}
\end{center}

\begin{center} {\footnotesize
$\tilde{P}'=$
\begin{prooftree}
\Infer{0}[ax]{ f:\,((2\cdot \alpha)\rightarrow \alpha)_2
 \trck{1}}

\Infer{0}[ax]{f:\,((2\cdot \alpha)\rightarrow \alpha)_4
\trck{1}}

\Hypo{\tilde{P}'}
\Ellipsis{}{f:\,((2\cdot \alpha)\rightarrow \alpha)_{k=3 \vee 
k\geqslant 5}\vdash f^\omega:\, \alpha \trck{2}}

\Infer{2}[app]{ f:\,((2\cdot \alpha)\rightarrow \alpha)_{k\geqslant 4}\vdash
f^\omega:\, \alpha \trck{2} }

\Infer{2}[app]{ f:\,((2\cdot \alpha)\rightarrow \alpha)_{k\geqslant 2}\vdash
\fom:\, \alpha}
\end{prooftree}}
\end{center}

Thus, $P'$ and $\tilde{P}'$ both represent $\Pi'$ (from subsec.
\ref{subsecExamples}), but $P'$ is quantitative whereas $\tilde{P}'$ is
not (the track $3$ w.r.t. $f$ does not end in an axiom leaf). Thus,
quantitativity is not stable under s.c.r.s. 

Moreover, it is easy to check that $P$ and $P'$ approximable (reuse the 
finite derivations of Subsection \ref{subsecExamples}). Thus, $\Pi$ and 
$\Pi'$ have both approximable and not approximable approximations. It 
provides a new argument for the impossibility of formulating 
approximability in system $\mathscr{M}$.

\end{document}